\documentclass[11pt]{article}
\usepackage{graphicx,amsmath,amsfonts,amssymb,amsthm,bbold,bm,dsfont,authblk,cite,xcolor,fullpage,braket,tikz}
\usepackage[colorlinks]{hyperref}
\usepackage[capitalize]{cleveref}
\usepackage[title]{appendix}
\usepackage[labelfont=bf]{caption}
\usepackage{thmtools} %for restatable theorems

\crefformat{equation}{(#2#1#3)}
\Crefformat{equation}{(#2#1#3)}
\crefname{section}{}{\S\S}
\crefformat{section}{#2\S#1#3}
\Crefformat{section}{#2\S#1#3}
\crefname{appendix}{Appendix}{Appendices}
\Crefname{appendix}{Appendix}{Appendices}

\definecolor{DarkGreen}{rgb}{0.1,0.5,0.1}
\definecolor{DarkBlue}{rgb}{0.1,0.1,0.7}
\definecolor{NiceOrange}{rgb}{.9,0.55,0}
\definecolor{NicePurple}{rgb}{0.3,0.1,1}
\hypersetup{
    unicode=false,          % non-Latin characters in Acrobat's bookmarks
    pdftoolbar=true,        % show Acrobat toolbar?
    pdfmenubar=true,        % show Acrobat menu?
    pdffitwindow=false,      % page fit to window when opened
    pdfnewwindow=true,      % links in new window
    colorlinks=true,       % false: boxed links; true: colored links
    linkcolor=DarkBlue,          % color of internal links
    citecolor=NicePurple,        % color of links to bibliography
    filecolor=DarkGreen,      % color of file links
    urlcolor=DarkBlue,          % color of external links
    pdftitle={},
    pdfauthor={},
}
\newcommand{\tr}{\mathrm{tr}}

\renewcommand{\th}{^\mathrm{th}}
\newcommand{\hobj}{H_{\mathrm{obj}}}
\newcommand{\id}{\mathds{1}}

\newtheorem{definition}{Definition}
\newtheorem{theorem}{Theorem}
\newtheorem{corollary}{Corollary}
\newtheorem{lemma}{Lemma}
\newtheorem{proposition}{Proposition}

\newif\ifnotes
\notestrue

\title{Optimization Using Locally-Quantum Decoders}

\author[1]{Noah Shutty\footnote{shutty@google.com}}
\author[1,2]{Avijit Mandal}
\author[3]{Seyoon Ragavan}
\author[4]{Quentin Buzet}
\author[4]{Andr\'e Chailloux}
\author[1]{Nicholas C.~Rubin}
\author[5]{Abid Khan}
\author[5]{Sami Boulebnane}
\author[5]{Ruslan Shaydulin}
\author[6]{John Azariah}
\author[1]{Stephen P.~Jordan}

\affil[1]{\small{\it{Google Quantum AI, Venice, CA 90291}}}
\affil[2]{\small{\it{Department of Electrical and Computer Engineering, Duke University, Durham, NC 27708}}}
\affil[3]{\small{\it{Computer Science and Artificial Intelligence Laboratory, MIT, Cambridge, MA 02139}}}
\affil[4]{\small{\it{Inria, 48 rue Barrault 75013, Paris, France}}}
\affil[5]{\small{\it{Global Technology Applied Research, JPMorganChase, New York, NY 10017}}}
\affil[6]{\small{\it{University of Technology Sydney, Sydney, Australia}}}

\date{}

\begin{document}

\maketitle

\begin{abstract}
It was pointed out in \cite{JSW25} that widely-studied optimization problems such as $D$-regular max-$k$-XORSAT can be reduced to decoding of LDPC codes, using quantum algorithms related to Regev's reduction. LDPC codes have very good decoders, such as Belief Propagation (BP), and this therefore makes $D$-regular max-$k$-XORSAT an enticing target for this class of quantum algorithms. However, BP was found insufficient to achieve quantum advantage. Here, we develop an intrinsically quantum decoding technique, which decodes classical LDPC codes subject to coherent superpositions of bit flip errors. For average-case instances of $D$-regular max-$k$-XORSAT drawn from Gallager's ensemble, this quantum decoder strongly outperforms classical belief propagation at many values of $k$ and $D$. For some $(k,D)$ the approximate optima achievable using this decoder surpass both Prange's algorithm and simulated annealing. However, we stop short of achieving quantum advantage because we identify an enhancement to Prange's algorithm that recovers a precise tie, much as a precise tie was observed between the standard version of Prange's algorithm and a more limited version of locally-quantum decoding in \cite{chailloux2023quantum}.
\end{abstract}

\section{Introduction}
\label{sec:intro}

Whether quantum computers can achieve exponential speedup for combinatorial optimization problems has been a subject of substantial research \cite{farhi2001quantum, farhi2014quantum, montanaro2020quantum, li2025new}. Many combinatorial optimization problems are NP-hard, and thus finding exact optima efficiently with quantum computers is impossible unless $\mathrm{NP} \subset \mathrm{BQP}$. However, it would not violate any standard assumptions about complexity classes if quantum computers can efficiently find approximate optima that require exponential time to find classically. Recently, an apparent exponential quantum speedup for an algebraically structured optimization problem called Optimal Polynomial Intersection (OPI) was obtained using a quantum algorithm called Decoded Quantum Interferometry (DQI) \cite{JSW25}, which descends from Regev's reduction \cite{ATS03,R04,AR05,R09}.

The success of DQI for the OPI problem encourages grander ambitions of achieving exponential speedup for algebraically unstructured combinatorial optimization problems, which are less specialized and more commonly encountered in industrial settings. Here we consider the max-XORSAT problem: we are given $n$ binary variables and $m$ constraints. Each constraint specifies the value (zero or one) of the XOR of some subset of the variables. These constraints can be regarded as linear equations over $\mathbb{F}_2$. If a solution exists that satisfies all $m$ constraints then it can be found in polynomial time classically using Gaussian elimination over $\mathbb{F}_2$. However, if no such solution exists, finding the assignment to the variables that maximizes the number of satisfied constraints is NP-hard.

It is common to study max-$k$-XORSAT, which is the special case of max-XORSAT where each constraint contains at most $k$ variables. If each variable is contained in exactly $D$ constraints and each constraint contains $k$ variables then the problem is called $D$-regular max-$k$-XORSAT. An instance of $D$-regular max-$k$-XORSAT can be expressed as finding the $\mathbf{x} \in \mathbb{F}_2^n$ that minimizes the objective function $f(\mathbf{x}) = |B \mathbf{x} - \mathbf{v}|$, where $\mathbf{v} \in \mathbb{F}_2^m$, $B \in \mathbb{F}_2^{m \times n}$ is a sparse matrix with $k$ nonzero entries per row and $D$ nonzero entries per column, and $|\cdot|$ denotes Hamming weight.

Although $D$-regular max-$k$-XORSAT is a somewhat idealized problem, it is widely used as a testbed for both classical and quantum optimization algorithms that exploit sparsity, which is a feature often seen in real-world optimization problems. The max-$k$-XORSAT problem remains NP-hard even when $k=2$. However, as $B$ becomes increasingly sparse, heuristic and exact solvers both classical and quantum are able to find approximate optima closer to the true optimum. See \textit{e.g.} \cite{BF21} and Section 9 of \cite{JSW25}.

Throughout this manuscript we consider average-case $D$-regular max-$k$-XORSAT instances where $B$ is drawn from Gallager's ensemble and $\mathbf{v}$ is chosen uniformly at random. Gallager's ensemble over $\mathbb{F}_2$-matrices with $D$ nonzero entries per column and $k$ nonzero entries per row is easy to sample from and widely used as a source of parity check matrices for classical LDPC codes. See \cref{app:trees} for a precise definition. The matrices from Gallager's ensemble have certain properties that generic $(k,D)$-sparse matrices lack, as described in \cref{sec:fgum}, which we exploit in both our quantum and classical algorithms for $D$-regular max-$k$-XORSAT.

Some arguments regarding the limitations of DQI, as applied to algebraically unstructured problems, are given in \cite{anschuetz2025decoded, parekh2025no, kramer2026tight}. None of these rule out quantum advantage by DQI for average-case $D$-regular max-$k$-XORSAT. The results of \cite{parekh2025no} show that DQI, when applied to MaxCut (a special case of max-2-XORSAT), using decoders restricted to the unique decoding radius of the resulting LDPC code, is provably outperformed by efficient classical optimization algorithms. This limitation does not apply here because we consider $k \geq 3$. Furthermore, in \cref{sec:errors} we bound the effect of decoding failures, and are therefore not restricted to the unique decoding radius. In \cite{kramer2026tight}, complexity-theoretic inapproximability results are given for max-LINSAT, which generalizes max-XORSAT to larger alphabet. These do not apply to the present work because they apply to worst-case instances with no bound on the ratio of constraints to variables, whereas here we consider average-case instances where this ratio is $D/k$ for various finite choices of $D$ and $k$. In \cite{anschuetz2025decoded}, it is conjectured that a formula from \cite{ZP75} for the number of errors correctable by a greedy bit-flipping decoder applied to LDPC codes is also an upper bound on the number of errors correctable by any efficient algorithm. The resulting performance of DQI is then shown to lie within the bounds on what can be achieved by algorithms whose output distribution is, in a certain formal sense, a stable function of the problem instance. However, the results in \cite{anschuetz2025decoded} are up to $o(1)$ corrections in the limit of large $k$, whereas here we are concerned with specific finite $k$. More fundamentally, the validity of this conjecture on the limitations of classical decoders is unknown, and one can find examples where quantum decoders exceed this bound and enable DQI to surpass the limitations that stable algorithms are subject to~\cite{zhou2026gibbs}.

DQI and Regev's reduction both use the quantum Fourier transform to reduce an instance of $D$-regular max-$k$-XORSAT to a problem of decoding an LDPC code in which each bit is contained in $k$ parity checks and each parity check contains $D$ bits. The matrix $B$ appearing in the objective function $f(\mathbf{x}) = |B \mathbf{x} - \mathbf{v}|$ yields the parity check matrix $B^T$ for the corresponding LDPC code. The relation between DQI and Regev's reduction is explained in \cref{app:DQI_vs_regev}. For reasons of technical convenience we work exclusively with Regev's reduction in this paper. However, the same techniques of locally-quantum decoding could be applied in the context of DQI.

The expected number of constraints that are satisfied by the bit strings produced by Regev's reduction is determined by the number of bit flip errors that can be decoded. The more bit flips decoded, the more constraints satisfied, as quantified in the context of Regev's reduction by \cref{eq:alphaperp} and in the context of DQI by the semicircle law \cite{JSW25}.

For codes defined by a generic dense parity check matrix $B^T$, there are no known polynomial-time decoding algorithms beyond error weight logarithmic in $n$. But for $B$ with $O(1)$ nonzero entries per row and column, linearly many errors can be efficiently corrected using known classical algorithms such as Belief Propagation (BP). Via the semicircle law this implies that DQI satisfies strictly more than half of the constraints in the limit $n \to \infty$. Thus, one might hope that DQI, using a reversible implementation of belief propagation as its decoder, could achieve quantum advantage over classical algorithms on $D$-regular max-$k$-XORSAT for at least some $(k,D)$ pairs.

The potential for DQI, with BP decoding, to achieve quantum advantage for average case $D$-regular max-$k$-XORSAT was explored in \cite{JSW25}. Two of the primary competing classical algorithms for average case $D$-regular max-$k$-XORSAT are simulated annealing and Prange's algorithm, with simulated annealing generally finding better optima for the smallest values of $k$ and $D$ (\textit{i.e.} the problems where $B$ is the sparsest), and Prange's algorithm finding better optima for denser problems. Via the semicircle law, the performance of DQI+BP can be inferred from the empirical performance of BP and compared against the empirical performance of simulated annealing and Prange's algorithm on the same instance. In \cite{JSW25}, no $(k,D)$ was found for which DQI+BP outperformed both simulated annealing and Prange's algorithm on $D$-regular max-$k$-XORSAT, although a highly tuned \textit{irregular} instance of max-XORSAT was constructed for which DQI+BP could outperform both simulated annealing and Prange's algorithm. This irregular instance was obtained by optimizing the irregular sparsity pattern to be maximally favorable to BP and maximally unfavorable to simulated annealing, and thus served as a proof of principle rather than as a practical example.

The primary competing quantum algorithm is the Quantum Approximate Optimization Algorithm (QAOA)~\cite{farhi2014quantum}. QAOA performance grows monotonically with QAOA depth $p$ and can be predicted classically for sparse max-$k$-XORSAT in time exponential in $p$ but independent of problem size for both the infinite-$D$ limit~\cite{BF21} and finite $D$ \cite{farhi2025lower}. We use the values from \cite{BF21} for the infinite-$D$ limit and extend the techniques previously introduced for $k=2$ in \cite{farhi2025lower} to arbitrary $k$ to calculate the finite $D$ performance.

As emphasized in \cite{chailloux2023quantum}, the bit flip errors that arise in DQI and related algorithms are in coherent superposition. Information-theoretically, a coherent superposition of such errors is more favorable than the corresponding probability distribution over classical bit flips, as the Holevo bound exceeds the Shannon bound. However, realizing this advantage with efficient (\textit{i.e.} polynomial size) quantum circuits is highly nontrivial and is the primary subject of this paper.

Here, following \cite{clz21,chailloux2023quantum,buzet2025fine}, we restrict our attention to what we call ``locally-quantum decoding schemes'' which, while in principle less powerful than general quantum decoders, are simpler to reason about and easier to implement by efficient quantum circuits. Roughly, the idea is to do local changes of basis on $O(1)$ qubits that modify the effective error channel to make it more advantageous, and then apply classical decoding to the resulting errors. Using a carefully chosen, yet relatively simple, version of locally-quantum decoding, we are able to consistently outcompete prior results achieved by DQI+BP, and for the first time outperform both simulated annealing and Prange's algorithm on $D$-regular max-$k$-XORSAT for some $(k,D)$, as shown in \cref{tab:scores}. The main new idea in our quantum decoder relative to the previous locally-quantum decoders introduced in \cite{clz21,chailloux2023quantum,buzet2025fine} is that our local change of basis is performed on sets of qubits corresponding to the parity checks in the underlying code and the \textit{a priori} knowledge that the uncorrupted codeword has parity zero on each of these sets is taken into account.

We do not claim outright quantum advantage because, as described in \cref{sec:turboprange}, we are able to construct a classical algorithm combining Prange with greedy local optimization that efficiently matches the approximate optima achieved by our quantum algorithm. Nevertheless, the fact that a relatively simple locally-quantum decoder can already greatly surpass DQI+BP and slightly surpass both simulated annealing and Prange's algorithm on a standard testbed problem suggests to us that the development of new quantum algorithms for the quantum decoding problem and the analysis of existing quantum algorithms for the quantum decoding problem such as BPQM \cite{renes2017belief,brandsen2022belief,mandal2026belief}, are among the most promising directions for the future development of DQI and Regev approaches to max-XORSAT.

  \begin{table}[!htbp]
       \centering
     \begin{tabular}{|c|c|c|c|c|c|}
     	\hline
     	$(k,D)$ & Prange  & Simulated          & DQI+BP  & Regev+FGUM  & QAOA       \\
     	        &         & Annealing          &         &             & ($p{=}16$)      \\
     	\hline
     	(3,4)   & 0.875   & \textbf{0.9366}    & 0.8730 & 0.8930      & 0.8898    \\
     	\hline
     	(3,5)   & 0.8     & \textbf{0.9005}    & 0.8176 & 0.8379      & 0.8532    \\
     	\hline
     	(3,6)   & 0.75    & \textbf{0.8712}    & 0.7776 & 0.7857      & 0.8231    \\
     	\hline
     	(3,7)   & 0.71428 & \textbf{0.8492}    & 0.7476 & 0.7621      &  0.8001   \\
     	\hline
     	(3,8)   & 0.6875  & \textbf{0.8287}    & 0.7243 & 0.7312      &  0.7813   \\
     	\hline
     	(4,5)   & 0.9     & \textbf{0.9279}    & 0.8605  & 0.9216      &   0.8797  \\
     	\hline    	
        (4,6)   & 0.83333 & \textbf{0.9024}    & 0.8214 & 0.8616      &  0.8498   \\
     	\hline
     	(4,7)   & 0.78571 & \textbf{0.8771}    & 0.7908 & 0.8267      &   0.8259  \\
     	\hline
     	(4,8)   & 0.75    & \textbf{0.8587}    & 0.7663 & 0.7905      &   0.8061  \\
     	\hline
     	(5,6)   & 0.91667 & 0.9190             & 0.8443 & \textbf{0.9312} & 0.8669\\
     	\hline
     	(5,7)   & 0.85714 & \textbf{0.8965}    & 0.8140 & 0.8853      &   0.8428  \\
     	\hline
     	(5,8)   & 0.8125  & \textbf{0.8740}    & 0.7893 & 0.8441       &  0.8226  \\
     	\hline
     	(6,7)   & 0.92857 & 0.9051             & 0.8291 & \textbf{0.9427} & 0.8546\\
     	\hline
     	(6,8)   & 0.875   & 0.8875             & 0.8045 & \textbf{0.8962}  &0.8344\\
     	\hline
     	(7,8)   & 0.9375  & 0.8155             & 0.8155   & \textbf{0.9481}  &0.8432\\
     	\hline
     \end{tabular}
   \caption{\label{tab:scores} Comparison of optimization for $D$-regular max-$k$-XORSAT. Top scores are in bold. Scores for Regev+FGUM are approximate up to corrections exponentially small in $D$. Regev+FGUM beats both Prange and DQI+BP for all $(k,D)$ in our table but is matched by Turbo Prange (see~\cref{sec:turboprange}). QAOA scores are computed exactly by evaluating and optimizing depth-$p$ QAOA on a $(D,k)$-regular hypertree (see~\cref{sec:QAOA}), and are accurate up to a $O(n^{-0.99})$ correction with high probability over the random choice of instance. The score monotonically increases with $p$. DQI+BP scores are obtained by computing the asymptotic behavior of BP using density evolution~\cite{RU01} and substituting into the semicircle law~\cite{JSW25}. Simulated annealing scores are obtained using instances with $n=2520$ variables and applying $10^6$ sweeps \textit{i.e.} $n \times 10^6$ Metropolis updates. For each instance, simulated annealing was run independently with four different seeds and the best score was recorded. Although for any given $(k,D)$ the approximate optimum found by simulated annealing can be slightly improved using more updates, this has diminishing returns. When $k$ and $D$ become large, Prange's algorithm beats simulated annealing by an increasing margin, and eventually simulated annealing ceases to be a viable competitor.}
 \end{table}
%simulated annealing numbers are from:
%/Users/stephenjordan/Documents/xortools/locally_quantum_data/second_sweep
%DQI+BP numbers are from bazel-bin/src/density, de.sh, de_results.txt, and dqi_bp_table.ipynb

\section{Regev's Reduction and the Quantum Decoding Problem}
\label{sec:quantum_decoding}

Although originally proposed as a reduction between different classical lattice problems to clarify their complexity and serve as a cryptographic foundation \cite{ATS03,R04,AR05,R09}, Regev's reduction can be adapted to reduce max-XORSAT to a quantum decoding problem, as is done in \cite{chailloux2023quantum}. We now review how this works. For simplicity, we will omit normalization factors when sketching the reduction in this section.

Recall that max-XORSAT is the problem of minimizing an objective of the form $f(\mathbf{x}) = |B \mathbf{x} - \mathbf{v}|$. This can also be viewed as the problem of finding the codeword $\mathbf{c} \in C$ of minimal Hamming distance from $\mathbf{v}$, where $C$ is the code defined by generator matrix $B$:
\begin{equation}
    \label{eq:C}
    C = \{ B \mathbf{x} | \mathbf{x} \in \mathbb{F}_2^n \}.
\end{equation}
To do this, one can try to prepare a state of the form
\begin{equation}
    \label{eq:final_state}
    \ket{\psi_P} \propto \sum_{\mathbf{x} \in \mathbb{F}_2^n} P(B\mathbf{x} - \mathbf{v}) \ket{B \mathbf{x}},
\end{equation}
where $P$ is a function that biases the superposition toward small $|B\mathbf{x} - \mathbf{v}|$. Upon measuring this state in the computational basis, one obtains a string of the form $B \mathbf{x}$. Inferring $\mathbf{x}$ from $B \mathbf{x}$ is then an overconstrained linear system over $\mathbb{F}_2$ that is guaranteed to have exactly one solution, provided $B$ is full rank. Furthermore, this solution can be found efficiently by Gaussian elimination. If the bias imposed by $P$ is strong then the resulting $\mathbf{x}$ will be a good solution to the original max-XORSAT problem.

The state $\ket{\psi_P}$ can be rewritten as
\begin{equation}
    \label{eq:final_state2}
    \ket{\psi_P} \propto \sum_{\mathbf{c} \in C} P(\mathbf{c} - \mathbf{v}) \ket{\mathbf{c}},
\end{equation}
where $C$ is the code defined in \cref{eq:C}. By the convolution theorem, the Hadamard transform of $\ket{\psi_P}$ is
\begin{equation}
    \label{eq:tilde_state}
    \ket{\widetilde{\psi}_P} \propto \sum_{\mathbf{d} \in C^\perp} \sum_{\mathbf{e} \in \mathbb{F}_2^m} (-1)^{\mathbf{v} \cdot \mathbf{e}} \widetilde{P}(\mathbf{e}) \ket{\mathbf{d} + \mathbf{e}},
\end{equation}
where $\widetilde{P}$ is the Hadamard transform of $P$ and $C^\perp$ is the code defined by parity check matrix $B^T \in \mathbb{F}_2^{n \times m}$. That is,
\begin{equation}
    \label{eq:Cperp}
    C^\perp = \{ \mathbf{d} \in \mathbb{F}_2^m : B^T \mathbf{d} = \mathbf{0} \}.
\end{equation}
Since the Hadamard transform is self-inverse, this suggests an approach to preparing $\ket{\psi_P}$: first prepare $\ket{\widetilde{\psi}_P}$ and then apply a Hadamard transform.

The final remaining problem is how to prepare the state $\ket{\widetilde{\psi}_P}$. In Regev's reduction one approaches this as follows. First prepare a uniform superposition over $C^\perp$ in one register, and a single copy of the peak defined by $\widetilde{P}$ in a second register.
\begin{equation}
    \left( \frac{1}{\sqrt{|C^\perp|}} \sum_{\mathbf{d} \in C^\perp} \ket{\mathbf{d}} \right) \otimes \sum_{\mathbf{e} \in \mathbb{F}_2^m} \widetilde{P}(\mathbf{e}) \ket{\mathbf{e}}.
\end{equation}
Preparing a uniform superposition over a linear code can be done by simple and efficient quantum circuits, as shown \textit{e.g.} in \cite{FJ24}. Preparing the peak $\widetilde{P}$ is efficient for many choices of $\widetilde{P}$, such as the ones we discuss in \cref{sec:locally_quantum} and \cref{sec:fgum}.

Second, reversibly add the content of the first register to the second register.
\begin{equation}
        \to \frac{1}{\sqrt{|C^\perp|}} \sum_{\mathbf{d} \in C^\perp} \ket{\mathbf{d}} \sum_{\mathbf{e} \in \mathbb{F}_2^m} \widetilde{P}(\mathbf{e}) \ket{\mathbf{d} + \mathbf{e}}.
\end{equation}
We now almost have $\ket{\widetilde{\psi}_P}$. Imposing the necessary phases $(-1)^{\mathbf{e} \cdot \mathbf{v}}$ is straightforward but we are faced with the nontrivial problem of uncomputing the unwanted register containing $\mathbf{d}$. In other words, one wishes to infer the codeword $\mathbf{d}$ from the corrupted codeword $\sum_{\mathbf{e} \in \mathbb{F}_2^m} \widetilde{P}(\mathbf{e}) \ket{\mathbf{d} + \mathbf{e}}$. This is the quantum decoding problem. It is information-theoretically possible to solve this if the superposition over error strings $\mathbf{e}$ defined by $\widetilde{P}(\mathbf{e})$ is narrowly peaked around low Hamming weight errors. In particular, if the support of $\widetilde{P}$ is limited to Hamming radius smaller than half the minimum Hamming distance between codewords of $C^\perp$ then the states $\sum_{\mathbf{e} \in \mathbb{F}_2^m} \widetilde{P}(\mathbf{e}) \ket{\mathbf{d} + \mathbf{e}}$ for different $\mathbf{d} \in C^\perp$ are orthogonal and perfect decoding is possible in principle. However, doing such decoding using efficient quantum circuits is highly nontrivial and is the central topic of this paper. A second core component of this paper is bounding the effect of decoding failures in the case that the superpositions $\sum_{\mathbf{e} \in \mathbb{F}_2^m} \widetilde{P}(\mathbf{e}) \ket{\mathbf{d} + \mathbf{e}}$ for different $\mathbf{d}$ are not orthogonal but only nearly so, \textit{e.g.} when $\widetilde{P}$ has support on all bit strings.

\section{Locally-Quantum Decoding}
\label{sec:locally_quantum}

For simplicity, suppose we had a quantum decoding problem for a single codeword $\mathbf{d}$ rather than a superposition over quantum decoding problems, one for each codeword. That is, suppose we are given the state $\sum_{\mathbf{e} \in \mathbb{F}_2^m} \widetilde{P}(\mathbf{e}) \ket{\mathbf{d} + \mathbf{e}}$ with the promise that $\mathbf{d} \in C^\perp$ and we wish to infer $\mathbf{d}$.

Whether this can be done (even information-theoretically) depends on the choice of $P$, which in turn determines $\widetilde{P}$. The simplest and most commonly studied choice is $P(\mathbf{y}) = P_\alpha(\mathbf{y})$, where
\begin{equation}
    \label{eq:palpha}
    P_\alpha(\mathbf{y}) := \prod_{j=1}^m \sqrt{1-\alpha}^{1-y_j} \sqrt{\alpha}^{y_j}.
\end{equation}
Here $\alpha$ is chosen between zero and one and $y_j$ denotes the $j\th$ bit of $\mathbf{y}$. In this case, if one had $C = \mathbb{F}_2^m$ then, by \cref{eq:final_state2}, upon measuring $\ket{\psi_P}$ in the computational basis, each of the $m$ constraints would be violated with probability $\alpha$, independently. Thus the expected number of constraints satisfied would be $(1-\alpha) m$. In reality, $C$ consists of only $2^n$ bit strings from $\mathbb{F}_2^m$. Nevertheless, for average case $\mathbf{v}$ the naive estimate that the expected number of satisfied constraints is $(1-\alpha) m$ turns out to be exactly correct. This is reflected in the $\epsilon=0$ special case of \cref{cor:expectation_bound2}.

For $\alpha \leq 1/2$, the Hadamard transform of $P_\alpha$ is
\begin{equation}
    \label{eq:tildepalpha}
    \widetilde{P}_\alpha(\mathbf{e}) = \prod_{j=1}^m \sqrt{1-\alpha^\perp}^{1-e_j} \sqrt{\alpha^\perp}^{e_j},
\end{equation}
where
\begin{equation}
    \label{eq:alphaperp}
    \alpha^\perp = \frac{1}{2} - \sqrt{\alpha (1-\alpha)}.
\end{equation}
As discussed above, $1-\alpha$ is the expected fraction of optimization constraints satisfied by the string obtained from measuring $\ket{\psi_P}$ in the computational basis. By \cref{eq:tilde_state}, $\alpha^\perp$ is the bit flip error rate on $\mathbf{d}$ if one were to measure $\ket{\widetilde{\psi}_P}$ in the computational basis. As discussed in \cref{sec:quantum_decoding}, we prepare the state $\ket{\widetilde{\psi}_P}$ by solving the decoding problem of inferring the original codeword $\mathbf{d}$ after it has been subjected to this error. Thus, one can interpret \cref{eq:alphaperp} as a relation between the effective error rate corrected in the decoding step, namely $\alpha^\perp$, and the fraction of constraints satisfied in the original optimization problem, namely $1-\alpha$. This is analogous to the role that the semicircle law plays in DQI \cite{JSW25}.

One way to infer $\mathbf{d} \in C^\perp$ given the state $\sum_{\mathbf{e} \in \mathbb{F}_2^m} \widetilde{P}(\mathbf{e}) \ket{\mathbf{d} + \mathbf{e}}$ would be to measure the state in the computational basis and try to classically infer $\mathbf{d}$ from the measured value of $\mathbf{d} + \mathbf{e}$. Due to the product structure of $\widetilde{P}_\alpha$, shown in \cref{eq:tildepalpha}, our quantum decoding problem is one of inferring the codeword $\mathbf{d}$ after each bit in $\mathbf{d}$ has been sent separately through the following classical-quantum channel.
\begin{equation}
    \label{eq:quantum_channel}
    d_j \to \sqrt{1-\alpha^\perp} \ket{d_j} + \sqrt{\alpha^\perp} \ket{\lnot d_j}
\end{equation}
By measuring in the computational basis, we collapse this to a classical bit flip channel where each bit is flipped with probability $\alpha^\perp$. This is called the binary symmetric channel $\mathrm{BSC}(\alpha^\perp)$. Inferring the original codeword $\mathbf{d}$ from the bit string $\mathbf{d} + \mathbf{e}$ resulting from the binary symmetric channel is one of the most-studied problems in classical error correction. In the case that the parity check matrix $B^T$ defining the code $C^\perp$ is sufficiently sparse, this problem can be solved with high probability in polynomial time using the belief propagation algorithm, even when the error rate $\alpha^\perp$ is close to saturating the information-theoretic limit on classical decodability, namely Shannon's bound.

In the context of Regev's reduction, one cannot actually measure $\ket{\mathbf{d} + \mathbf{e}}$ because one must solve a superposition over quantum decoding problems for different codewords $\mathbf{d} \in C^\perp$. Instead, one must implement the chosen classical decoding algorithm, such as belief propagation, as a reversible circuit and apply it to the superposition. The measurements serve only as a conceptual tool to aid us in thinking about the quantum decoding problem.

By using the conceptual tool of fictitious measurements one can perceive alternative approaches to the quantum decoding problem, which go outside the framework of simply leveraging existing classical decoders and applying them as reversible circuits. In particular, in \cite{chailloux2023quantum}, Unambiguous State Discrimination (USD) measurements were introduced in the context of Regev reductions. Given two non-orthogonal states $\ket{\psi_0}$ and $\ket{\psi_1}$, the corresponding USD measurement is a three-outcome POVM that takes $\ket{\psi_b}$ as input for unknown $b \in \{0, 1\}$ and will either identify with certainty that $b = 0$, identify with certainty that $b = 1$, or produce the $\perp$ outcome which indicates no information regarding the value of $b$. A larger inner product $|\langle \psi_0 | \psi_1 \rangle|$ necessitates a larger probability on the $\perp$ outcome. Specifically, in the case of uniform priors, the optimal USD measurement achieves
\begin{equation}
    p_{\perp} = |\langle \psi_0 | \psi_1 \rangle|.
\end{equation}
Hence, if one applies a USD measurement to the states produced by the classical-quantum channel of \cref{eq:quantum_channel} then instead of a classical binary symmetric channel with bit flip probability $\alpha^\perp$ one obtains a classical erasure channel with bit erasure probability $2 \sqrt{\alpha^\perp (1 - \alpha^\perp)}$.

Information-theoretically, the bit flip channel $\mathrm{BSC}(\alpha^\perp)$ has higher capacity than this erasure channel. However, the advantage of an erasure channel is that it is easier to saturate the information-theoretic capacity using polynomial-time decoders. Specifically, one can find the erased bit values by solving the set of $\mathbb{F}_2$-linear equations defined by the parity check matrix using Gaussian elimination. For decoding bit flip errors, known polynomial-time decoders fall increasingly short of information-theoretic limits as the density of the parity check matrix increases. Decoding erasures using Gaussian elimination does not suffer from this problem.

The imperviousness of erasure decoders to density is enticing because the optimization problems that are reducible to quantum decoding by Regev's reduction become much more difficult to approximate using local search heuristics such as simulated annealing when the instances become denser. By the principle of deferred measurement \cite{NC00}, anything that can be achieved by measurement and classical feedforward can also be achieved coherently by controlled reversible operations. Thus, instead of carrying out the USD measurements, one can dilate them to isometries, implemented concretely by introducing ancilla qubits each initialized to $\ket{0}$ and applying unitary transformations. Then, by applying a reversible implementation of Gaussian elimination, one can replicate the performance of the USD-based decoding method in a purely coherent way applicable to the quantum decoding problem on superpositions of codewords, as needed in Regev's reduction. Unfortunately, as shown in \cite{chailloux2023quantum}, the approximate optima found by applying this strategy are exactly matched by Prange's algorithm, which is a classical algorithm that runs in polynomial time.

The USD-based strategy from \cite{chailloux2023quantum} is one example of what we here call a locally-quantum approach to the quantum decoding problem. One considers a choice of local measurement that induces a classical channel and finds a classical decoding algorithm suited to the resulting channel. The measurement-inspired scheme is then implemented in a fully coherent manner by replacing the measurements with isometries. If the measurement is in the computational basis, we regard the approach as fully classical, since the coherent version consists of nothing other than a reversible implementation of a classical decoding algorithm. The USD-based strategy is not the only locally-quantum approach to the quantum decoding problem that has been introduced. Prior to \cite{chailloux2023quantum}, a method based on filtering measurements was introduced in \cite{clz21}. Subsequent to \cite{chailloux2023quantum} a method based on unambiguous state discrimination for groups of qubits was introduced in \cite{buzet2025fine}. Quantum advantage via Regev's reduction was not found in \cite{chailloux2023quantum} or \cite{buzet2025fine}, as the approximate optima could be efficiently replicated classically using Prange's algorithm. A potential quantum advantage was found in \cite{clz21}. However, a polynomial-time classical algorithm replicating or exceeding the approximate optima achieved by the quantum algorithm of \cite{clz21} was subsequently found \cite{kothari2025no}. Thus, at present, there is no known quantum advantage achieved by using locally-quantum approaches to the quantum decoding problem in the context of Regev's reduction.

\section{Fine-Grained Unambiguous Measurements (FGUM)}
\label{sec:fgum}

Here, we consider a new locally-quantum decoding method based on Fine-Grained Unambiguous Measurements (FGUM), which were introduced in \cite{buzet2025fine}. The analysis in \cite{buzet2025fine} showed that Regev's reduction using locally-quantum decoding based on FGUMs could not beat Prange's algorithm. However, the measurements considered in \cite{buzet2025fine} are code-agnostic. In this section, we extend FGUMs to include measurements that depend on the code. We find that these more tailored FGUMs do enable our quantum algorithm to beat the satisfaction fractions of both Prange's algorithm and simulated annealing for some problems.

This is achieved by careful codesign of the function $P$ used to bias the output state toward good solutions and the measurements used in the locally-quantum decoding process. The first step in this codesign is to partition the $m$ qubits corresponding to the code symbols into groups of size $D$ that are aligned with the constraint structure of the given instance of $D$-regular max-$k$-XORSAT. We specifically consider instances of $D$-regular max-$k$-XORSAT drawn from Gallager's ensemble, as more general ensembles of instances pose additional technical challenges, which we do not address here.

The structure of an instance of max-XORSAT can be represented using a bipartite graph. One row of $n$ vertices represents the $n$ binary variables and one row of $m$ vertices represents the $m$ constraints. Each constraint vertex is connected by an edge to each variable vertex that it contains. This is also the Tanner graph of the code $C^\perp$ dual to the max-XORSAT problem. When viewed as a Tanner graph, the $m$ vertices represent the $m$ bits in $C^\perp$ and the $n$ vertices represent the $n$ parity checks in $C^\perp$. For an instance of $D$-regular max-$k$-XORSAT this bipartite graph will be $(k,D)$-regular; each variable vertex will connect to exactly $D$ edges and each constraint vertex will connect to exactly $k$ edges. A convenient and commonly used way to generate random $(k,D)$-regular bipartite graphs is to draw them from Gallager's ensemble, as described \textit{e.g.} in appendix B of \cite{JSW25}. 

\begin{figure}[htbp]
    \centering
    \resizebox{0.65\textwidth}{!}{%
        \begin{tikzpicture}[
            checknode/.style={circle, draw=black, fill=white, inner sep=2pt, outer sep=0pt},
            varnode/.style={circle, draw=black, fill=white, inner sep=2pt, outer sep=0pt},
            thinedge/.style={draw=gray, thin},
            thickedge/.style={draw=blue!80!black, line width=2pt},
            labeltext/.style={font=\sffamily\small, align=right} 
        ]
            \def\hspacing{1.5}
            \def\vspacing{2}
            \node[varnode] (B1) at (-2.5*\hspacing, 0) {};
            \node[varnode] (B2) at (-1.5*\hspacing, 0) {};
            \node[varnode] (B3) at (-0.5*\hspacing, 0) {};
            \node[varnode] (B4) at ( 0.5*\hspacing, 0) {};
            \node[varnode] (B5) at ( 1.5*\hspacing, 0) {};
            \node[varnode] (B6) at ( 2.5*\hspacing, 0) {};
            \node[checknode, fill=blue!80!black] (T1) at (-1.5*\hspacing, \vspacing) {}; 
            \node[checknode] (T2) at (-0.5*\hspacing, \vspacing) {};
            \node[checknode] (T3) at ( 0.5*\hspacing, \vspacing) {};
            \node[checknode, fill=blue!80!black] (T4) at ( 1.5*\hspacing, \vspacing) {}; 
            \node[labeltext, anchor=east] at (-2.9*\hspacing, \vspacing) {check nodes};
            \node[labeltext, anchor=east] at (-2.9*\hspacing, 0) {bit nodes};
            \draw[thinedge] (T2) -- (B1);
            \draw[thinedge] (T2) -- (B2);
            \draw[thinedge] (T2) -- (B4);
            \draw[thinedge] (T3) -- (B3);
            \draw[thinedge] (T3) -- (B5);
            \draw[thinedge] (T3) -- (B6);
            \draw[thinedge] (T1) -- (B1);
            \draw[thinedge] (T1) -- (B2);
            \draw[thinedge] (T1) -- (B3);
            \draw[thinedge] (T4) -- (B4);
            \draw[thinedge] (T4) -- (B5);
            \draw[thinedge] (T4) -- (B6);
            \draw[thickedge] (T1) -- (B1);
            \draw[thickedge] (T1) -- (B2);
            \draw[thickedge] (T1) -- (B3);
            \draw[thickedge] (T4) -- (B4);
            \draw[thickedge] (T4) -- (B5);
            \draw[thickedge] (T4) -- (B6);
        \end{tikzpicture}%
    }
    \caption{\label{fig:gallager_tanner} A Tanner graph from the $(2,3)$ Gallager ensemble. Highlighted in blue is a set of parity check nodes such that each bit node is connected by an edge to exactly one check from this set.}
\end{figure}
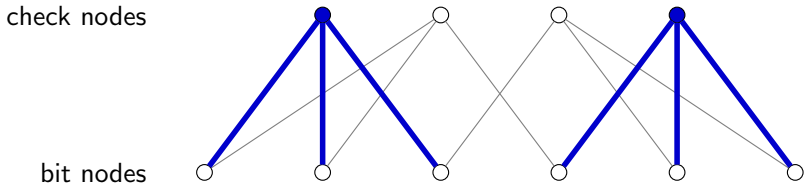

In any LDPC code drawn from Gallager's ensemble, there exists a subset $\Gamma \subset \{1,\ldots,n\}$ of parity checks such that each bit is contained in exactly one parity check from this subset\footnote{If the sample from Gallager's ensemble is generated in the standard way, one can simply use the first $m/D$ rows from the parity check matrix.}. Hence $|\Gamma| = m/D$. This is illustrated in \cref{fig:gallager_tanner}. Although general $(k,D)$-regular bipartite graphs do not have this property, we consider only graphs that do. Any such choice of $\Gamma$ defines a partition of the $m$ bits of the code $C^\perp$ into blocks $S_1,\ldots,S_{m/D} \subset \{1,\ldots,m\}$. $S_i$ is the set of bits in the $i\th$ parity check from $\Gamma$. Thus, $S_1,\ldots,S_{m/D}$ are disjoint and $\bigcup_{i=1}^{m/D} S_i = \{1,\ldots,m\}$ and $|S_i|=D$ for all $i=1,\ldots,m/D$.

Building upon this partition, we next describe our specific choice of $P$ that defines our output state $\ket{\psi_P}$ via \cref{eq:final_state}. For any bit string $\mathbf{y} \in \mathbb{F}_2^m$ and any subset $S \subset \{1,\ldots,m\}$ let $\mathbf{y}_S$ denote the corresponding length-$|S|$ substring. Our choice of $P$ is of the form
\begin{equation}
    \label{eq:prodP}
    P(\mathbf{y}) = \prod_{i=1}^{m/D} P^{(D)}(\mathbf{y}_{S_i}).
\end{equation}
Consequently, its Hadamard transform is of the form
\begin{equation}
    \label{eq:prodPtilde}
    \widetilde{P}(\mathbf{e}) = \prod_{i=1}^{m/D} \widetilde{P}^{(D)}(\mathbf{e}_{S_i}).
\end{equation}

To describe $P^{(D)}$ we need the following definitions.
\begin{definition}
	Let $D$ be an even integer. A set $J_D$ is called a symmetric split of balanced words of size $D$ if
	\begin{itemize}
		\item $J_D \subseteq \{\mathbf{y} \in \mathbb{F}_2^D : |\mathbf{y}| = \frac{D}{2}\}$
		\item $|J_D| = \frac{1}{2} \cdot \binom{D}{D/2}$
		\item $J_D \cup (J_D + 1^D) = \{\mathbf{y} \in \mathbb{F}_2^D : |\mathbf{y}| = \frac{D}{2}\}$.
	\end{itemize} 
\end{definition}
\vspace{11pt}
\begin{definition}
    \label{def:ID}
	For each integer $D$, we fix a set $I_D \subseteq \mathbb{F}_2^D$ with $|I_D| = 2^{D-1}$ that minimizes the quantity $\sum_{\mathbf{y} \in I_D} |\mathbf{y}|$ and such that $I_D \cup (I_D + 1^D) = \mathbb{F}_2^D$. More precisely:
	\begin{itemize}
		\item If $D$ is odd, we choose $I_D = \{\mathbf{y} \in \mathbb{F}_2^D : |\mathbf{y}| \le (D-1)/2\}$.
		\item If $D$ is even, we choose 
		$I_D = \{\mathbf{y} \in \mathbb{F}_2^D : |\mathbf{y}| \le \frac{D}{2} - 1\} \cup J_D$, where $J_D$ is any symmetric split of balanced words of size $D$.
	\end{itemize}
	We also write $I_D^* = I_D \backslash \{0^D\}$.
\end{definition}
We remark that the set $I_D$ corresponds to the ``coset leaders'' of \cite{buzet2025fine}.
With these definitions in hand, we can define $P^{(D)}$ as follows
\begin{equation}
    \label{eq:PDdef}
    P^{(D)}(\mathbf{y}_{S_i}) = \left\{ \begin{array}{cl}
	\sqrt{1-\alpha} & \textrm{if } \mathbf{y}_{S_i} = 0^D \\
	\sqrt{\frac{\alpha}{2^{D-1} - 1}} & \textrm{if } \mathbf{y}_{S_i} \in I_D^* \\
	0 & \textrm{otherwise.}
	\end{array}\right.
\end{equation}

Using this choice of $P$, we obtain a superposition over quantum decoding problems, one for each $\mathbf{d} \in C^\perp$. By \cref{eq:prodPtilde}, our quantum decoding problem is, given a state of the form
\begin{equation}
    \label{eq:prodp2}
    \sum_{\mathbf{e} \in \mathbb{F}_2^m} \prod_{i=1}^{m/D} \widetilde{P}^{(D)}(\mathbf{e}_{S_i}) \ket{\mathbf{d} + \mathbf{e}},
\end{equation}
to recover the codeword $\mathbf{d} \in C^\perp$.

We now describe our decoding method, which we call Fine-Grained Unambiguous Measurements (FGUM), following \cite{buzet2025fine}. Up to ordering of qubits, the state described in \cref{eq:prodp2} can be re-expressed as
\begin{equation}
    \sum_{\mathbf{e} \in \mathbb{F}_2^m} \widetilde{P}(\mathbf{e}) \ket{\mathbf{d} + \mathbf{e}} = \bigotimes_{i=1}^{m/D} \ket{\widetilde{\psi}_i^{(D)}}
\end{equation}
where
\begin{equation}
    \ket{\widetilde{\psi}_i^{(D)}} = \sum_{\mathbf{e} \in \mathbb{F}_2^D} \widetilde{P}^{(D)}(\mathbf{e}) \ket{\mathbf{d}_{S_i} + \mathbf{e}}.
\end{equation}
We measure each of these $m/D$ tensor factors as follows. First, we apply a Hadamard transform, yielding $H^{\otimes D} \ket{\widetilde{\psi}_i^{(D)}} = \ket{\psi_i^{(D)}}$ given by
\begin{equation}
    \label{eq:psied}
    \ket{\psi_i^{(D)}} = \sum_{\mathbf{y} \in \mathbb{F}_2^D} (-1)^{\mathbf{d}_{S_i} \cdot \mathbf{y}} P^{(D)}(\mathbf{y}) \ket{\mathbf{y}}.
\end{equation}
Next, we perform the two-outcome measurement $\{ M_0, M_1 \}$ where
\begin{align}
    M_0 &= \left(\frac{\alpha}{(1-\alpha)(2^{D-1} - 1)}\right)^{1/2} \ket{0^D}\bra{0^D} + \sum_{\mathbf{y} \neq 0^D} \ket{\mathbf{y}}\bra{\mathbf{y}} \label{eq:M0}\\
    M_1 &=  \left(1-\frac{\alpha}{(1-\alpha)(2^{D-1} - 1)}\right)^{1/2} \ket{0^D}\bra{0^D}. \label{eq:M1}
\end{align}
This measurement is non-projective but satisfies the postulates of a general measurement in quantum mechanics, provided $\alpha \leq 1-2^{1-D}$. As we will later show, this condition is satisfied by all values of $\alpha$ used in our algorithm. (See \cref{eq:alphamin}.)

From this measurement, we obtain the outcome $0$ with probability
\begin{equation}
    p_0 = \bra{\psi_i^{(D)}} M_0^\dag M_0 \ket{\psi_i^{(D)}}.
\end{equation}
By \cref{eq:M0} and \cref{eq:psied},
\begin{equation}
    p_0 =  \frac{\alpha}{(1-\alpha)(2^{D-1} - 1)} |P^{(D)}(0^D)|^2 + \sum_{\mathbf{y} \neq 0^D} |P^{(D)}(\mathbf{y})|^2.
\end{equation}
Hence, by \cref{eq:PDdef}
\begin{equation}
    \label{eq:p0alpha}
    p_0  = \alpha \frac{2^{D-1}}{2^{D-1} - 1}.
\end{equation}
Conditioned on getting the $0$ outcome, the post-measurement state is
$\frac{M_0 \ket{\psi_i^{(D)}}}{\| M_0 \ket{\psi_i^{(D)}} \|}$. By \cref{eq:M0}, \cref{eq:psied}, and \cref{eq:PDdef}, 
\begin{equation}
    \frac{M_0 \ket{\psi_i^{(D)}}}{\| M_0 \ket{\psi_i^{(D)}} \|} = \frac{1}{\sqrt{2^{D-1}}}\sum_{\mathbf{y} \in I_D} (-1)^{\mathbf{d}_{S_i} \cdot \mathbf{y}} \ket{\mathbf{y}}.
\end{equation}
Next, we perform the operation $\ket{\mathbf{i}} \rightarrow \frac{1}{\sqrt{2}} \left( \ket{\mathbf{i}} + \ket{\mathbf{i} + 1^D}\right)$ for each $\mathbf{i} \in I_D$. This can be done unitarily due to the definition of $I_D$. Since $\mathbf{d}_{S_i} \cdot 1^D = 0$, we thereby obtain the state 
\begin{equation}
    \frac{1}{\sqrt{2^D}}\sum_{\mathbf{y} \in \mathbb{F}_2^D} (-1)^{\mathbf{d}_{S_i} \cdot \mathbf{y}} \ket{\mathbf{y}}.
\end{equation}
Applying a Hadamard transform and measuring in the computational basis then yields $\mathbf{d}_{S_i}$.

We have described the above procedure as a sequence of two measurements. We first apply a measurement with outcomes 0 and 1. If we get outcome 1 we give up. If we get outcome 0 we then apply a second measurement with outcomes in $\mathbb{F}_2^D$. We could also conceptualize this as a single measurement with $2^D+1$ outcomes: one ``failure'' outcome which yields no information about $\mathbf{d}_{S_i}$, and $2^D$ ``success'' outcomes which tell us $\mathbf{d}_{S_i}$ with certainty. In describing and analyzing the Regev+FGUM algorithm throughout the remainder of this manuscript we will use this latter formulation of FGUM as a measurement with $2^D+1$ outcomes.

After applying the FGUM to each of the $m/D$ blocks of qubits from our partition, we can postprocess the resulting classical information by regarding the measurement outcomes, along with the \textit{a priori} knowledge that $B^T \mathbf{d} = \mathbf{0}$ as a combined system of linear equations over $\mathbb{F}_2$. If this is fully determined, we obtain $\mathbf{d}$. Next, we analyze with what probability this occurs.

\section{Analysis of FGUM}
\label{sec:FGUM_analysis}

We can interpret the failed measurements as erasure errors. There are $m/D$ blocks in our partition. On each one we apply a fine-grained unambiguous measurement, which succeeds with probability $p_0$. Thus the error model seen in the classical postprocessing of the measurement outcomes is one in which each block is erased with probability $1-p_0$. In addition to the measurement outcomes, we also have the \textit{a priori} knowledge that $B^T \mathbf{d} = \mathbf{0}$. This is a system of $n$ equations. Each erased bit is a variable that we need to solve for. Thus, one might naively expect that the classical decoder can succeed (using Gaussian elimination) as long as the total expected number of erasures, namely $(1-p_0)m$, is less than or equal to $n$. However, this is not quite correct, due to some linear dependencies between the equations. We thus define $e_{\max}$ as follows.

\begin{definition}
    Let $B \in \mathbb{F}_2^{m \times n}$ be a matrix from the $(k,D)$-Gallager ensemble and let $C^\perp = \{ \mathbf{d} \in \mathbb{F}_2^m : B^T \mathbf{d} = \mathbf{0} \}$ be the corresponding LDPC code. Let $S_1,\ldots,S_{m/D}$ be a partition of the codeword bits derived from $B$ as in \cref{sec:fgum}. We define $e_{\max}$ as the maximum probability such that if each block of bits from a codeword $\mathbf{d} \in C^\perp$ is independently erased with probability $e_{\max}$ the erased bits can be recovered with high probability in the limit $n \to \infty$ by solving the system of $\mathbb{F}_2$-linear equations arising from the parity check condition $B^T \mathbf{d} = \mathbf{0}$.
\end{definition}

Let $\nu$ be the number of erased bits (out of $m$) that we need to solve for. Let's now estimate $E$, the expected number of $\mathbb{F}_2$-linear equations constraining these unknown bits. We expect the linear system to become solvable when $E=\nu + O(1)$. We have a total of $n$ parity check equations. Some of these equations contain only known bits. $E$ counts the equations that contain unknown bits.

To calculate $E$ we divide the $n$ parity check equations into three categories. The first is a parity check equation that defines a block from the partition that has been erased. In expectation there are $E_1 = \nu/D$ such equations. All of these equations contain erased bits. The second is a parity check equation that defines a block from the partition that has not been erased. In expectation there are $(m-\nu)/D$ of these. However, none of these contain any erased bits so their contribution to $E$ is $E_2 = 0$. The third is a parity check equation that does not define a block from the partition. There are $n-m/D$ of these. Some of these contain erased bits but not all. Each of these equations contains $D$ bits. Since $\nu$ out of the $m$ bits have been erased, one estimates that the probability that none of these $D$ bits in such a parity check are erased is $(1-\nu/m)^D$. Hence, from these $n-m/D$, the expected number containing erased bits is $E_3 \simeq (n-m/D)(1-(1-\nu/m)^D)$. Our estimate of the total number of equations constraining the erased bits is $E = E_1+E_2+E_3$. That is,
\begin{equation}
    \label{eq:E}
    E \simeq \frac{\nu}{D} + \left(n - \frac{m}{D} \right) \left( 1 - \left(1-\frac{\nu}{m} \right)^D \right).
\end{equation}
To solve for $e_{\max}$ we substitute $\nu = e_{\max} m$ into \cref{eq:E} and demand that the number of equations equals the number of unknowns, \textit{i.e.} $E = e_{\max} m$. That is,
\begin{equation}
    \label{eq:emax1}
    e_{\max} \simeq \frac{e_{\max}}{D} + \left( \frac{n}{m} - \frac{1}{D} \right) \left( 1 - \left(1-e_{\max} \right)^D \right).
\end{equation}
For Gallager's ensemble $n/m = k/D$, and thus \cref{eq:emax1} is more conveniently written as
\begin{equation}
    \label{eq:emax}
    e_{\max} \simeq \frac{e_{\max}}{D} + \left( \frac{k}{D} - \frac{1}{D} \right) \left( 1 - \left(1-e_{\max} \right)^D \right).
\end{equation}
We can solve for the largest root of \cref{eq:emax} numerically to find our approximation to $e_{\max}$ for any given $k$ and $D$.

To find the minimum achievable $\alpha$, we set $e_{\max} = 1 - p_0$ and use \cref{eq:p0alpha}. This yields
\begin{equation}
    \label{eq:alphamin}
    \alpha_{\min} = (1-e_{\max}) (1-2^{1-D}).
\end{equation}

Lastly, we consider the expected number of constraints satisfied in our max-XORSAT instance as a function of $\alpha$. From \cref{eq:PDdef}, we see that the expected number of satisfied constraints in a block of $D$ is
\begin{equation}
    \label{eq:bsda}
    \bar{s}(D, \alpha) = (1-\alpha) D + \alpha \bar{I}^*_D,
\end{equation}
where $\bar{I}^*_D$ is the expected number of constraints satisfied by a uniformly random element of $I^*_D$. Observe that 
\begin{equation}
    \label{eq:bisd}
    \bar{I}^*_D = D-\hat{I}^*_D,
\end{equation}
where $\hat{I}^*_D$ is the expected Hamming weight of a uniformly random element of $I^*_D$. Substituting \cref{eq:bisd} into \cref{eq:bsda} and simplifying yields
\begin{equation}
    \label{eq:exp}
    \bar{s}(D,\alpha) = D - \alpha \hat{I}^*_D.
\end{equation}
By \cref{def:ID},
\begin{equation}
    \label{eq:hatid}
    \hat{I}_D^* = \left\{ \begin{array}{cl} \frac{\sum_{w=1}^{(D-1)/2} \binom{D}{w} w}{\sum_{w=1}^{(D-1)/2} \binom{D}{w}} & \textrm{if $D$ is odd} \vspace{10pt} \\
                                          \frac{\sum_{w=1}^{D/2-1} \binom{D}{w} w + \frac{1}{2} \binom{D}{D/2} \frac{D}{2}}{\sum_{w=1}^{D/2-1} \binom{D}{w} + \frac{1}{2} \binom{D}{D/2}}
                                         & \textrm{if $D$ is even.} \end{array} \right.
\end{equation}
The expected fraction of satisfied constraints produced by the quantum algorithm is then
\begin{equation}
    \label{eq:satfrac}
    \frac{\langle s \rangle}{m} = \frac{\bar{s}(D,\alpha_{\min})}{D}.
\end{equation}
By substituting \cref{eq:alphamin} and \cref{eq:exp} into \cref{eq:satfrac} and simplifying one obtains
\begin{equation}
    \label{eq:simple_satfrac}
    \frac{\langle s \rangle}{m} = 1 - \left( 1 - e_{\max} \right) \left( 1 - 2^{1-D} \right) \frac{\hat{I}^*_D}{D}.
\end{equation}
Evaluating \cref{eq:hatid} and \cref{eq:emax} and substituting the results into \cref{eq:simple_satfrac} yields the satisfaction fractions shown in the Regev+FGUM column of \cref{tab:scores}.

\section{Error Bound}
\label{sec:errors}

In the previous section, we calculated the maximum satisfaction fraction achievable such that our FGUM method is expected to succeed on the resulting average case quantum decoding problem. However, this does not mean that the decoding will succeed with $100\%$ probability for each codeword in the superposition. Rather, the failure probability $\epsilon$ will in general take some nonzero value, which depends on our choice of $\alpha$. In this section, we prove that if $\epsilon$ is small, then Regev's reduction will still work well solving the max-XORSAT problem $\min_{\mathbf{x}} |B \mathbf{x} - \mathbf{v}|$, for average case $\mathbf{v}$. In fact, with an eye toward future work, we prove an error bound as a function of $\epsilon$ that is much more general than what we need for our immediate purpose and then obtain a bound for Regev's reduction applied to max-XORSAT as a corollary. Many similar error bounds exist in the literature \cite{chailloux2024quantum, JSW25, quantumEquivalenceSLWEISIS}, but none of the existing theorems perfectly fits our present needs.

We first establish some notation and terminology in order to state our main result in its full generality. Throughout this section, we consider optimization and decoding problems over the finite field $\mathbb{F}_q$ of characteristic $p$. The dual pair of max-XORSAT and binary decoding arise in the special case of $\mathbb{F}_2$. We start by formally defining the quantum decoding problem for a general alphabet.
\begin{definition}
    \label{def:QDP} The quantum decoding problem $\mathrm{QDP}_{C^\perp,\widetilde{P}}$ is the following task: given as input the state $\sum_{\mathbf{e} \in \mathbb{F}_q^m} \widetilde{P}(\mathbf{e}) \ket{\mathbf{d} + \mathbf{e}}$, for a uniformly random $\mathbf{d} \in C^\perp$, output $\mathbf{d}$.
\end{definition}

\begin{definition}\label{def:qdpstinespring}
    By Stinespring's dilation theorem~\cite{NC00}, we may assume without loss of generality that any algorithm for $\mathrm{QDP}_{C^\perp, \widetilde{P}}$: (a) appends an ancilla register $\ket{0_\mathrm{anc}}$ in the all 0's state; (b) applies a unitary $U_\mathrm{dec}$; then (c) measures the final register in the standard basis.
    We write the action of $U_\mathrm{dec}$ as:
    \begin{equation}
    \label{eq:udec}
    U_{\mathrm{dec}} \sum_{\mathbf{e} \in \mathbb{F}_q^m} \widetilde{P}(\mathbf{e}) \ket{\mathbf{d} + \mathbf{e}} \ket{0_\mathrm{anc}} = \sqrt{1-\epsilon_{\mathbf{d}}} \ket{\xi_{\mathbf{d}}} \ket{\mathbf{d}} + \sqrt{\epsilon_{\mathbf{d}}} \ket{\delta_{\mathbf{d}}}.
    \end{equation}
    Here, $\ket{\xi_{\mathbf{d}}}$ is an arbitrary normalized garbage state. $\ket{\delta_{\mathbf{d}}}$ is a normalized state of both registers which has zero amplitude on any state in which the second register contains $\mathbf{d}$, but is otherwise arbitrary.
    The probability that this algorithm recovers $\mathbf{d}$ successfully is therefore $1-\epsilon_\mathbf{d}$.
    The overall success probability of the $\mathrm{QDP}$ algorithm is then defined as $1 - \frac{1}{|C^\perp|} \sum_{\mathbf{d} \in C^\perp} \epsilon_\mathbf{d}$.
\end{definition}

\noindent
For any $a \in \mathbb{F}_q$, let:
\begin{equation}
    \label{eq:shift_single}
    X_q^{(a)} = \sum_{j \in \mathbb{F}_q} \ket{j+a}\bra{j}
\end{equation}
be the shift operator for $\mathbb{F}_q$. If $\mathbb{F}_q$ is a prime field then $X_q^{(a)}$ is simply $X_q^{(1)}$ raised to the $a\th$ power, but for extension fields this is not true. For $\mathbb{F}_2$ we recover the Pauli $X$ operator. Given $\mathbf{v} \in \mathbb{F}_q^m$ define
\begin{equation}
    \label{eq:shift_multi}
    X_q^{\mathbf{v}} = X_q^{(v_1)} \otimes \cdots \otimes X_q^{(v_m)}.
\end{equation}
Next, we state the main result of this section.
\begin{restatable}[Main error bound]{theorem}{mainErrorBound}
    \label{thm:expectation_bound} Let $\hobj \in \mathbb{C}^{q^m \times q^m}$ be a diagonal matrix with entries in $[0,1]$ and $\ket{\mathcal{S}} = \sum_{\mathbf{x} \in \mathbb{F}_q^m} P(\mathbf{x}) \ket{\mathbf{x}}$ be a normalized state. Let $\widetilde{P}$ denote the Fourier transform of $P$ over $\mathbb{F}_q^m$. Let $C$ be a code over $\mathbb{F}_q^m$ and let $C^\perp$ be its dual. Using any quantum algorithm for $\mathrm{QDP}_{C^\perp,\widetilde{P}}$ succeeding with probability $1-\epsilon$, Regev's reduction produces a mixed state $\rho(\mathbf{v})$ with support only on the codewords of $C$ such that
    \[
    \mathbb{E}_{\mathbf{v}} \left[ \tr \left[ X_q^{\mathbf{v}} \hobj X_q^{-\mathbf{v}} \rho(\mathbf{v}) \right]\right] \geq \braket{\mathcal{S} | \hobj | \mathcal{S}} - 2 \sqrt{\epsilon},
    \]
    where the expectation is taken over uniform $\mathbf{v} \in \mathbb{F}_q^m$.
\end{restatable}
Note: if $q$ is not prime then the Fourier transform over $\mathbb{F}_q^m$ is defined using the field trace, as discussed in~\cref{sec:error_proofs_prelims}. For $q=2$ this reduces to the ordinary Hadamard transform.

To apply this to $D$-regular max-$k$-XORSAT we take our field to be $\mathbb{F}_2$ and our dual pair of codes $C,C^\perp$ to be as defined in \cref{eq:C} and \cref{eq:Cperp}. We then choose
\begin{equation}
    \hobj = \sum_{\mathbf{x} \in \mathbb{F}_2^m} \left( 1 - \frac{|\mathbf{x}|}{m}\right) \ket{\mathbf{x}} \bra{\mathbf{x}}.
\end{equation}
We thus obtain the following as a special case of \cref{thm:expectation_bound}.
\begin{corollary}
    \label{cor:expectation_bound2}
    Given $B \in \mathbb{F}_2^{m \times n}$, let $C=\{B \mathbf{x} : \mathbf{x} \in \mathbb{F}_2^n \}$ and $C^\perp = \{\mathbf{d} \in \mathbb{F}_2^m : B^T \mathbf{d} = \mathbf{0}\}$. Let $P:\mathbb{F}_2^m \to \mathbb{C}$ be a function satisfying the normalization condition $\sum_{\mathbf{x} \in \mathbb{F}_2^m} |P(\mathbf{x})|^2 = 1$. Let $\widetilde{P}$ denote the Hadamard transform of $P$. Choose $\mathbf{v} \in \mathbb{F}_2^m$ uniformly at random, and apply Regev's reduction to the max-XORSAT problem of finding a codeword $\mathbf{c} \in C$ minimizing $|\mathbf{c} - \mathbf{v}|$. Using any quantum algorithm for $\mathrm{QDP}_{C^\perp,\widetilde{P}}$ succeeding with probability $1-\epsilon$, this procedure produces a codeword $\mathbf{c} \in C$ such that
    \[
        \mathbb{E}_{\mathbf{v}} \left[ 1-\frac{|\mathbf{c} - \mathbf{v}|}{m} \right] \geq \left[ \sum_{\mathbf{y} \in \mathbb{F}_2^m} |P(\mathbf{y})|^2 \left( 1- \frac{|\mathbf{y}|}{m} \right) \right] - 2 \sqrt{\epsilon}.
    \]
\end{corollary}
Informally, this says that, for max-XORSAT with average case $\mathbf{v}$, a nonzero failure rate $\epsilon$ for the quantum decoding reduces the fraction of constraints satisfied by the bit string produced by Regev's reduction by at most $2 \sqrt{\epsilon}$.

\section{Proof of Error Bound}
\label{sec:error_proofs}

\subsection{Preliminaries}\label{sec:error_proofs_prelims}
In preparation for proving \cref{thm:expectation_bound}, we first introduce some terminology, notation, and useful theorems. All of the theorems in this subsection are well known and can be proven by direct calculation, which we omit.

\begin{definition}[Characters]
    Let $\omega_p := \exp(2\pi i/p)$. For an element $a \in \mathbb{F}_q$ we define the trace $\tr(a) := \sum_{i = 0}^{r-1} a^{p^i} \in \mathbb{F}_p$, where $r$ is such that $q = p^r$. 
    For elements $a, b \in \mathbb{F}_q$ the character is:
    \begin{equation}
        \chi_a(b) = \omega_p^{\tr(ab)}.
    \end{equation}
    For vectors $\mathbf{a}, \mathbf{b} \in \mathbb{F}_q^m$, we define $\chi_{\mathbf{a}}(\mathbf{b}) = \prod_{i = 1}^m \chi_{a_i}(b_i)$.
\end{definition}
Some important properties of characters are the homomorphic property
\begin{equation}
    \label{eq:character_homomorphic}
    \chi_{\mathbf{a}}(\mathbf{c}_1) \chi_{\mathbf{a}}(\mathbf{c}_2) = \chi_\mathbf{a}(\mathbf{c}_1+\mathbf{c}_2),
\end{equation}
the orthogonality property
\begin{equation}
    \label{eq:character_orthogonality}
    \sum_{\mathbf{a} \in \mathbb{F}_q^m} \chi_{\mathbf{a}}(\mathbf{c}) = \left\{ \begin{array}{ll} q^m & \textrm{if $\mathbf{c} = \mathbf{0}$} \\ 0 & \textrm{otherwise} \end{array}\right.
\end{equation}
and the complex conjugate identity
\begin{equation}
    \label{eq:character_cc}
    \chi_{\mathbf{a}}(\mathbf{c})^* = \chi_{\mathbf{a}}(-\mathbf{c}).
\end{equation}

Using the quantum Fourier transform over $\mathbb{F}_q$ we can define clock operators to go with the shift operators defined in \cref{eq:shift_single} and \cref{eq:shift_multi}, as follows. Let $Z_q^{(a)}$ denote the clock operator:
\begin{equation}
    Z_q^{(a)} = \sum_{x \in \mathbb{F}_q} \chi_a(x)\ket{x} \bra{x}.
\end{equation}
For a vector $\mathbf{v} \in \mathbb{F}_q^m$ we also use the following notation, analogous to \cref{eq:shift_multi}:
\begin{equation}
    Z_q^{\mathbf{v}} = Z_q^{(v_1)} \otimes \cdots \otimes Z_q^{(v_m)}.
\end{equation}
For $\mathbb{F}_2$ the clock and shift operators reduce to the standard Pauli $Z$ and $X$ operators. Using the field trace, we next define the quantum Fourier transform over $\mathbb{F}_q$. 
\begin{definition}[Quantum Fourier transform]\label{def:QFT}
    Let $\mathrm{QFT}_q$ denote the unitary such that
    \begin{equation}
        \mathrm{QFT}_q \ket{x} = \frac{1}{\sqrt{q}}\sum_{y\in \mathbb{F}_q} \chi_x(y) \ket{y}.
    \end{equation}
    Furthermore, for a quantum state $\ket{\psi} = \sum_x \psi_x \ket{x}$ we denote
    \begin{equation}
        \ket{\widetilde{\psi}} = \sum_{x}\widetilde{\psi}_x\ket{x} = \mathrm{QFT}_q \ket{\psi}.
    \end{equation}
\end{definition}
\noindent
Note that we slightly abuse notation and denote by $\mathrm{QFT}$ any of $\mathrm{QFT}_q^{\otimes j}$ for some $j>0$ (the field size $q$ will be a variable or clear from context).

With these definitions in hand it is straightforward to show that
\begin{equation}\label{eq:QFTXZdual}
    \mathrm{QFT} \ X_q^{(a)} \ \mathrm{QFT}^\dag = Z_q^{(a)} \quad \mathrm{and} \quad \mathrm{QFT}^\dag \ X_q^{(a)} \ \mathrm{QFT} = Z_q^{(-a)}
\end{equation}
and
\begin{equation}
    \mathrm{QFT}_q^\dag \ket{x} = \frac{1}{\sqrt{q}}\sum_{y\in \mathbb{F}_q} \chi_x(-y) \ket{y}.
\end{equation}
\begin{definition}[Convolution of quantum states]
For quantum states $\ket{\psi} = \sum_{\mathbf{x} \in \mathbb{F}_q^m} \psi_{\mathbf{x}} \ket{\mathbf{x}}$ and $\ket{\phi} = \sum_{\mathbf{x} \in \mathbb{F}_q^m} \phi_{\mathbf{x}} \ket{\mathbf{x}}$, their convolution is the unnormalized quantum state
\begin{equation}
    \ket{\psi} * \ket{\phi} = \ket{\psi * \phi} = \sum_{\mathbf{x},\mathbf{y} \in \mathbb{F}_q^m} \psi_{\mathbf{x}} \phi_{\mathbf{y}} \ket{\mathbf{x}+\mathbf{y}}.
\end{equation}
\end{definition}
\begin{definition}[Pointwise product of quantum states]\label{def:pointwise_product}
For quantum states $\ket{\psi} = \sum_{\mathbf{x} \in \mathbb{F}_q^m} \psi_{\mathbf{x}} \ket{\mathbf{x}}$ and $\ket{\phi} = \sum_{\mathbf{x} \in \mathbb{F}_q^m} \phi_{\mathbf{x}} \ket{\mathbf{x}}$, their pointwise product is the unnormalized quantum state
\begin{equation}
    \ket{\psi}\odot\ket{\phi} = \ket{\psi \odot \phi} = \sum_{\mathbf{x} \in \mathbb{F}_q^m}\psi_{\mathbf{x}} \phi_{\mathbf{x}} \ket{\mathbf{x}}.
\end{equation}
\end{definition}
\noindent
\begin{theorem}[Fourier Convolution Theorem]\label{thm:convolution}
For quantum states $\ket{\psi} = \sum_{\mathbf{x} \in \mathbb{F}_q^m} \psi_{\mathbf{x}} \ket{\mathbf{x}}$ and $\ket{\phi} = \sum_{\mathbf{x} \in \mathbb{F}_q^m} \phi_{\mathbf{x}} \ket{\mathbf{x}}$,
    \begin{equation}
        \mathrm{QFT} \ket{\psi \odot \phi} = \frac{1}{\sqrt{q^m}} \left(\mathrm{QFT} \ket{\psi}\right) * \left(\mathrm{QFT} \ket{\phi}\right).
    \end{equation}
    This identity also holds if we replace $\mathrm{QFT}$ with $\mathrm{QFT}^\dag$ throughout.
\end{theorem}
\begin{definition}[Code States]\label{def:code_states}
    Given a linear code $C \subset \mathbb{F}_q^m$, the corresponding code state $\ket{C}$ is the normalized uniform superposition over codewords in $C$:
    \[
        \ket{C} = \frac{1}{\sqrt{|C|}}\sum_{\mathbf{c} \in C} \ket{\mathbf{c}}.
    \]
\end{definition}
\begin{theorem}\label{thm:subspaceQFT}
    Given $B \in \mathbb{F}_q^{m \times n}$, let $C$ be the code $C=\{B\mathbf{x}: \mathbf{x} \in \mathbb{F}_q^n \}$ and let $C^\perp$ be its dual code $C^\perp = \{\mathbf{d} \in \mathbb{F}_q^m : B^T \mathbf{d} = \mathbf{0}\}$. Then
    \[
    \mathrm{QFT} \ket{C} = \ket{C^\perp} \quad \textrm{and} \quad \mathrm{QFT}^\dag \ket{C} = \ket{C^\perp}.
    \]
\end{theorem}
\noindent
Additionally, we use the notation $\mathbf{1}[\textrm{condition}]$ as the indicator function, which evaluates to 1 when the condition is met and 0 otherwise. 

\subsection{Regev's Reduction over a General Finite Field}

Before proving \cref{thm:expectation_bound}, we must first review what we mean by Regev's reduction over a general finite field $\mathbb{F}_q$. \\

\noindent
\textbf{Step 1: Prepare the initial state.} In the first register prepare the uniform superposition over the code $C^\perp$. Since $C^\perp$ is a linear code this can be done by an efficient quantum circuit. In the second register prepare a superposition with amplitudes given by $\widetilde{P}$. In this paper, $\widetilde{P}$ factors into a product of functions, each acting on only $D$ symbols. Consequently, this state is a tensor product of states on $D$ qudits, which can therefore be prepared efficiently for any constant $D$.
\begin{equation}
    \left( \frac{1}{\sqrt{|C^\perp|}} \sum_{\mathbf{d} \in C^\perp } \ket{\mathbf{d}} \right) \otimes \left( \sum_{\mathbf{e} \in \mathbb{F}_q^m } \widetilde{P}(\mathbf{e}) \ket{\mathbf{e}} \right)
\end{equation}
\textbf{Step 2: Apply $Z_q^{-\mathbf{v}}$ to the first register.}
\begin{equation}
    \to \left( \frac{1}{\sqrt{|C^\perp|}} \sum_{\mathbf{d} \in C^\perp } \chi_{\mathbf{v}}(-\mathbf{d}) \ket{\mathbf{d}} \right) \otimes \left( \sum_{\mathbf{e} \in \mathbb{F}_q^m} \widetilde{P}(\mathbf{e}) \ket{\mathbf{e}} \right)
\end{equation}
\textbf{Step 3: Reversibly add the first register's contents to the second register.}
\begin{equation}
    \to  \frac{1}{\sqrt{|C^\perp|}} \sum_{\mathbf{d} \in C^\perp } \chi_{\mathbf{v}}(-\mathbf{d}) \ket{\mathbf{d}} \sum_{\mathbf{e} \in \mathbb{F}_q^m } \widetilde{P}(\mathbf{e}) \ket{\mathbf{d} + \mathbf{e}}
\end{equation}
\textbf{Step 4: Uncompute the first register.} This is done by solving the quantum decoding problem as specified in~\cref{def:qdpstinespring}.
We append qubits in the initial state $\ket{0_\mathrm{anc}}$ and then apply $U_\mathrm{dec}$ as defined in~\cref{eq:udec}.
Since our decoder fails with probability $\epsilon$ we see from \cref{def:qdpstinespring} that
\begin{equation}
    \label{eq:epsilon_avg}
    \frac{1}{|C^\perp|} \sum_{\mathbf{d} \in C^\perp} \epsilon_{\mathbf{d}} = \epsilon.
\end{equation}
Appending the ancilla qubits and applying $U_\mathrm{dec}$ at this stage in the algorithm yields
\begin{equation}
    \to \frac{1}{\sqrt{|C^\perp|}} \sum_{\mathbf{d} \in C^\perp} \chi_{\mathbf{v}}(-\mathbf{d}) \ket{\mathbf{d}} \left( \sqrt{1 - \epsilon_{\mathbf{d}}} \ket{\xi_{\mathbf{d}}} \ket{\mathbf{d}} + \sqrt{\epsilon_{\mathbf{d}}} \ket{\delta_{\mathbf{d}}} \right).
\end{equation}
\textbf{Step 5: Postselect on decoding success.} To do this we perform a projective measurement to see whether the last register contains the same symbol string as the first. If not, we abort and restart from step 1. This happens with probability $\epsilon$. If the measurement finds that the last register does contain the same symbol string as the first we are left with the state
\begin{equation}
    \label{eq:after_postselection}
    \to \mathcal{N}_{\mathrm{dec}} \sum_{\mathbf{d} \in C^\perp} \chi_{\mathbf{v}}(-\mathbf{d}) \ket{\mathbf{d}} \sqrt{1-\epsilon_{\mathbf{d}}} \ket{\xi_{\mathbf{d}}} \ket{\mathbf{d}},
\end{equation}
where $\mathcal{N}_{\mathrm{dec}} \in \mathbb{R}^{> 0}$ is a normalization constant. \\
\textbf{Step 6: Uncompute the first register.} This is achieved by reversibly subtracting the content of the last register from the content of the first register (over $\mathbb{F}_q^m$). The first register is then in the all zeros state, unentangled with the last register, and can therefore be safely discarded. The result is
\begin{equation}
    \to \mathcal{N}_{\mathrm{dec}} \sum_{\mathbf{d} \in C^\perp} \chi_{\mathbf{v}}(-\mathbf{d}) \sqrt{1-\epsilon_{\mathbf{d}}} \ket{\xi_{\mathbf{d}}} \ket{\mathbf{d}}.
\end{equation}
\textbf{Step 7: Apply $U_{\mathrm{dec}}^\dag$.} By \cref{eq:udec}, this yields
\begin{equation}
    \to \mathcal{N}_{\mathrm{dec}} \sum_{\mathbf{d} \in C^\perp} \chi_{\mathbf{v}}(-\mathbf{d}) \left( \sum_{\mathbf{e} \in \mathbb{F}_q^m} \widetilde{P}(\mathbf{e}) \ket{\mathbf{d} + \mathbf{e}} \ket{0_\mathrm{anc}} - \sqrt{\epsilon_{\mathbf{d}}} U_{\mathrm{dec}}^\dag \ket{\delta_{\mathbf{d}}} \right).
\end{equation}
\textbf{Step 8: Apply $Z_q^{\mathbf{v}}$ to the first register.}
\begin{align}
    &\to \mathcal{N}_{\mathrm{dec}} \sum_{\mathbf{d} \in C^\perp} \chi_{\mathbf{v}}(-\mathbf{d}) \left( \sum_{\mathbf{e} \in \mathbb{F}_q^m} \widetilde{P}(\mathbf{e}) \chi_{\mathbf{v}}(\mathbf{d} + \mathbf{e}) \ket{\mathbf{d} + \mathbf{e}} \ket{0_\mathrm{anc}} - \sqrt{\epsilon_{\mathbf{d}}} (Z_q^{\mathbf{v}} \otimes \id_\mathrm{anc}) U_{\mathrm{dec}}^\dag \ket{\delta_{\mathbf{d}}} \right) \\
    &= \mathcal{N}_{\mathrm{dec}} \sum_{\mathbf{d} \in C^\perp} \sum_{\mathbf{e} \in \mathbb{F}_q^m} \chi_{\mathbf{v}}(\mathbf{e}) \widetilde{P}(\mathbf{e}) \ket{\mathbf{d} + \mathbf{e}} \ket{0_\mathrm{anc}} - \ket{\Delta},
\end{align}
where
\begin{equation}
    \label{eq:delta}
    \ket{\Delta} = \mathcal{N}_{\mathrm{dec}} \sum_{\mathbf{d} \in C^\perp} \chi_{\mathbf{v}}(-\mathbf{d}) \sqrt{\epsilon_{\mathbf{d}}} (Z_q^{\mathbf{v}} \otimes \id_\mathrm{anc}) U_{\mathrm{dec}}^\dag \ket{\delta_{\mathbf{d}}}.
\end{equation}
\textbf{Step 9: Apply the inverse QFT and measure the first register.} By the Fourier convolution theorem (\cref{thm:convolution}), the inverse quantum Fourier transform of the state from the preceding step is
\begin{equation}
    \to \mathcal{N}_{\mathrm{dec}}|C^\perp| \sum_{\mathbf{c} \in C} P(\mathbf{c}-\mathbf{v}) \ket{\mathbf{c}} \ket{0_\mathrm{anc}} - (\mathrm{QFT}^\dag \otimes \id_\mathrm{anc}) \ket{\Delta}.
\end{equation}
(See the next subsection for a more detailed derivation.) Thus, upon measuring the first register, if $\ket{\Delta}$ had zero norm, we would obtain the codeword $\mathbf{c}$ with probability proportional to $|P(\mathbf{c}-\mathbf{v})|^2$, as desired. Thus, our essential task is to show that $\| \ket{\Delta} \|$ is small, for average case $\mathbf{v}$. This is made nontrivial by the fact that the states $\ket{\delta_{\mathbf{d}}}$ need not be orthogonal for different $\mathbf{d} \in C^\perp$.

\subsection{Technical Lemmas}
\label{sec:lemmas}

The technical lemmas in this section, which generalize the analysis of~\cite{chailloux2024quantum}, provide us with a toolkit that we use to prove \cref{thm:expectation_bound}. We hope that this modular analytical framework will allow for easy adaptation to other algorithms and tasks.
Let
\begin{equation}
    \label{eq:phi_target}
    \ket{\Phi_{\mathrm{target}}} = \mathcal{N}_{\mathrm{target}} \sum_{\mathbf{d} \in C^\perp} \sum_{\mathbf{e} \in \mathbb{F}_q^m} \chi_{\mathbf{v}}(\mathbf{e}) \widetilde{P}(\mathbf{e}) \ket{\mathbf{d} + \mathbf{e}} \ket{0_\mathrm{anc}}
\end{equation}
denote the state that would be produced at the end of step 8 in the absence of a $\ket{\Delta}$ term due to decoding failure. $\mathcal{N}_{\mathrm{target}} \in \mathbb{R}^{> 0}$ is chosen such that this state has unit norm. Note that in general $\mathcal{N}_{\mathrm{target}}$ depends on $\mathbf{v}$. In the language of convolutions and pointwise products, this can be expressed as
\begin{equation}
    \ket{\Phi_{\mathrm{target}}} =  \left(\sqrt{|C^\perp|} \ \mathcal{N}_{\mathrm{target}} \ \ket{C^\perp} * \left( Z_q^{\mathbf{v}} \ket{\widetilde{P}}\right)\right) \otimes \ket{0_\mathrm{anc}}
\end{equation}
where
\begin{align}
    \ket{\widetilde{P}} &= \sum_{\mathbf{x} \in \mathbb{F}_q^m} \widetilde{P}(\mathbf{x}) \ket{\mathbf{x}},% \\
    %\ket{\chi_{\mathbf{v}}} &= \sum_{\mathbf{x} \in \mathbb{F}_q^m} \chi_{\mathbf{v}}(\mathbf{x}) \ket{\mathbf{x}},
\end{align}
and $\ket{C^\perp}$ is the code state as in \cref{def:code_states}. The inverse Fourier transforms are
\begin{align}
    \textrm{QFT}^\dag \ket{C^\perp} &= \ket{C} \quad \textrm{(by \cref{thm:subspaceQFT})} \\
    \textrm{QFT}^\dag \ket{\widetilde{P}} &= \ket{P}. \\
    %\textrm{QFT}^\dag \ket{\chi_{\mathbf{v}}} &= \ket{\mathbf{v}}.
\end{align}
Thus, by \cref{eq:QFTXZdual,thm:convolution}, $\ket{\widetilde{\Phi}_{\mathrm{target}}} := \left(\textrm{QFT}^\dag \otimes \id_\mathrm{anc}\right) \ket{\Phi_{\mathrm{target}}}$ is given by
\begin{equation}
    \label{eq:tildephi}
    \ket{\widetilde{\Phi}_{\mathrm{target}}} = \sqrt{\frac{q^m |C^\perp|}{|C|}} \mathcal{N}_{\mathrm{target}} \sum_{\mathbf{c} \in C} P(\mathbf{c}-\mathbf{v}) \ket{\mathbf{c}} \ket{0_\mathrm{anc}} = |C^\perp| \mathcal{N}_{\mathrm{target}} \sum_{\mathbf{c} \in C} P(\mathbf{c}-\mathbf{v}) \ket{\mathbf{c}} \ket{0_\mathrm{anc}}.
\end{equation}
This is the final state that we would ideally wish to produce.  

Let $\ket{\Phi_{\mathrm{actual}}}$ be the state obtained from step 8 using a quantum algorithm that solves the quantum decoding problem with error probability $\epsilon$. As described in the previous subsection,
\begin{equation}
    \label{eq:phi_actual}
    \ket{\Phi_{\mathrm{actual}}} = \mathcal{N}_{\mathrm{dec}} \sum_{\mathbf{d} \in C^\perp} \sum_{\mathbf{e} \in \mathbb{F}_q^m} \chi_{\mathbf{v}}(\mathbf{e}) \widetilde{P}(\mathbf{e}) \ket{\mathbf{d} + \mathbf{e}} \ket{0_\mathrm{anc}} - \ket{\Delta},
\end{equation}
where $\ket{\Delta}$ is given by \cref{eq:delta}. By the unitarity of the QFT, $\langle \Phi_{\mathrm{actual}}| \Phi_{\mathrm{target}} \rangle$ is equal to the inner product between the final output state from Regev's reduction, and the ideal final state $\ket{\widetilde{\Phi}_{\mathrm{target}}}$. This motivates the following lemma.

\begin{lemma}\label{lem:expectedinnerproduct}
    If $U_{\mathrm{dec}}$ has failure probability $\epsilon$ then for uniformly random $\mathbf{v} \in \mathbb{F}_q^m$, we have:
    \begin{equation}
        \mathbb{E}_{\mathbf{v}} \left[\frac{\braket{\Phi_\mathrm{actual} | \Phi_\mathrm{target}}}{\mathcal{N}_\mathrm{dec} \mathcal{N}_\mathrm{target}}\right] = |C^\perp| (1-\epsilon).
    \end{equation}
\end{lemma}
\begin{proof}
    By \cref{eq:phi_target}, \cref{eq:phi_actual}, and \cref{eq:delta},
    \begin{align*}
        \frac{\braket{\Phi_\mathrm{actual} | \Phi_\mathrm{target}}}{\mathcal{N}_\mathrm{dec} \mathcal{N}_\mathrm{target}}
        &= \sum_{\substack{\mathbf{d}_1, \mathbf{d}_2 \in C^\perp \\ \mathbf{e}_1, \mathbf{e}_2 \in \mathbb{F}_q^m}} \chi_{\mathbf{v}}(-\mathbf{e}_1) \widetilde{P}^*(\mathbf{e}_1) \chi_{\mathbf{v}}(\mathbf{e}_2) \widetilde{P}(\mathbf{e}_2) \cdot \mathbf{1}[\mathbf{d}_1 + \mathbf{e}_1 = \mathbf{d}_2 + \mathbf{e}_2] \\
        &- \sum_{\substack{\mathbf{d}_1, \mathbf{d}_2 \in C^\perp \\ \mathbf{e}_2 \in \mathbb{F}_q^m}} \chi_{\mathbf{v}}(\mathbf{d}_1) \sqrt{\epsilon_{\mathbf{d}_1}} \chi_{\mathbf{v}}(\mathbf{e}_2) \widetilde{P}(\mathbf{e}_2) \bra{\delta_{\mathbf{d}_1}} U_\mathrm{dec} \left((Z_q^{-\mathbf{v}} \ket{\mathbf{d}_2+\mathbf{e}_2}) \otimes \ket{0_\mathrm{anc}}\right).
    \end{align*}
    By observing that $Z_q^{-\mathbf{v}} \ket{\mathbf{d}_2 + \mathbf{e}_2} = \chi_{\mathbf{v}}(-\mathbf{d}_2 -\mathbf{e}_2) \ket{\mathbf{d}_2 + \mathbf{e}_2}$ and applying the homomorphic property of characters \cref{eq:character_homomorphic}, we simplify this to
    \begin{align*}
        = &\sum_{\substack{\mathbf{d}_1, \mathbf{d}_2 \in C^\perp \\ \mathbf{e}_1, \mathbf{e}_2 \in \mathbb{F}_q^m}} \chi_{\mathbf{v}}(\mathbf{e}_2-\mathbf{e}_1) \widetilde{P}^*(\mathbf{e}_1) \widetilde{P}(\mathbf{e}_2) \cdot \mathbf{1}[\mathbf{d}_1 + \mathbf{e}_1 = \mathbf{d}_2 + \mathbf{e}_2] \\
        - & \sum_{\substack{\mathbf{d}_1, \mathbf{d}_2 \in C^\perp \\ \mathbf{e}_2 \in \mathbb{F}_q^m}} \chi_{\mathbf{v}}(\mathbf{d}_1-\mathbf{d}_2) \sqrt{\epsilon_{\mathbf{d}_1}} \widetilde{P}(\mathbf{e}_2) \bra{\delta_{\mathbf{d}_1}} U_\mathrm{dec} \left(\ket{\mathbf{d}_2+\mathbf{e}_2} \otimes \ket{0_\mathrm{anc}}\right).
    \end{align*}
    It follows by the orthogonality property of characters \cref{eq:character_orthogonality} that:
    \begin{align*}
        \mathbb{E}_{\mathbf{v}} \left[\frac{\langle \Phi_\mathrm{actual} | \Phi_\mathrm{target} \rangle}{\mathcal{N}_\mathrm{dec} \mathcal{N}_\mathrm{target}}  \right] &= \sum_{\substack{\mathbf{d} \in C^\perp \\ \mathbf{e} \in \mathbb{F}_q^m}} | \widetilde{P}(\mathbf{e})|^2 
        - \sum_{\substack{\mathbf{d} \in C^\perp \\ \mathbf{e} \in \mathbb{F}_q^m}} \sqrt{\epsilon_{\mathbf{d}}} \widetilde{P}(\mathbf{e}) \bra{\delta_\mathbf{d}} U_\mathrm{dec} \left(\ket{ \mathbf{d}+\mathbf{e}} \otimes \ket{0_\mathrm{anc}}\right) \\
        &= \ |C^\perp| - \sum_{\substack{\mathbf{d} \in C^\perp \\ \mathbf{e} \in \mathbb{F}_q^m}} \sqrt{\epsilon_\mathbf{d}} \widetilde{P}(\mathbf{e}) \bra{\delta_\mathbf{d}} U_\mathrm{dec} \left(\ket{\mathbf{d}+\mathbf{e}} \otimes \ket{0_\mathrm{anc}}\right) \\
        &=\ |C^\perp| - \sum_{\substack{\mathbf{d} \in C^\perp}} \sqrt{\epsilon_\mathbf{d}}  \bra{\delta_\mathbf{d}} U_\mathrm{dec} \left(\sum_{\mathbf{e} \in \mathbb{F}_q^m} \widetilde{P}(\mathbf{e}) \ket{\mathbf{d}+\mathbf{e}} \otimes \ket{0_\mathrm{anc}}\right).
    \end{align*}
    By the definition of $U_\mathrm{dec}$ in~\cref{eq:udec} this is equal to
    \[
        =  \ |C^\perp| - \sum_{\mathbf{d} \in C^\perp} \epsilon_\mathbf{d},
    \]
    which by \cref{eq:epsilon_avg} is equal to
    \[
        = \ |C^\perp|(1-\epsilon).
    \]
\end{proof}

\begin{lemma}\label{lem:Ndec}
    If $\epsilon$ is the failure probability of $U_{\mathrm{dec}}$, then
    \begin{equation}
        \mathcal{N}_{\mathrm{dec}} = 1/\sqrt{|C^\perp|(1-\epsilon)}.
    \end{equation}
\end{lemma}
\begin{proof}
    By \cref{eq:after_postselection},
    \begin{align*}
        1 &= \mathcal{N}_{\mathrm{dec}}^2 \sum_{\mathbf{d} \in C^\perp} | \chi_{\mathbf{v}}(-\mathbf{d}) \sqrt{1 - \epsilon_\mathbf{d}}|^2 \\
        &= \mathcal{N}_{\mathrm{dec}}^2 \sum_{\mathbf{d} \in C^\perp} (1 - \epsilon_\mathbf{d}) \\
        &= \mathcal{N}_{\mathrm{dec}}^2 |C^\perp| (1-\epsilon),
    \end{align*}
    where the last equality follows from \cref{eq:epsilon_avg}.
\end{proof}
Unlike $\mathcal{N}_{\mathrm{dec}}$, $\mathcal{N}_{\mathrm{target}}$ depends on $\mathbf{v}$. To deal with this, we use the following lemma.
\begin{lemma}\label{lem:Ntarg}
    For uniformly random $\mathbf{v} \in \mathbb{F}_q^m$,
    \begin{equation}
        \mathbb{E}_{\mathbf{v}} \left[ \frac{1}{\mathcal{N}_{\mathrm{target}}^2}\right] = |C^\perp|.
    \end{equation}
\end{lemma}
\begin{proof}
    Using \cref{eq:phi_target} and discarding the ancilla qubits that are in the all 0's state in $\ket{\Phi_\mathrm{target}}$, we have:
    \begin{align*}
        1/\mathcal{N}_{\mathrm{target}}^2 &= \left\| \sum_{\mathbf{d} \in C^\perp} \sum_{\mathbf{e} \in \mathbb{F}_q^m} \chi_{\mathbf{v}}(\mathbf{e}) \widetilde{P}(\mathbf{e}) \ket{\mathbf{d} + \mathbf{e}} \right\|^2 \\
        &= \left\| \sum_{\mathbf{d} \in C^\perp} \sum_{\mathbf{x} \in \mathbb{F}_q^m} \chi_{\mathbf{v}}(\mathbf{x} - \mathbf{d}) \widetilde{P}(\mathbf{x} - \mathbf{d}) \ket{\mathbf{x}} \right\|^2 \\
        &= \sum_{\mathbf{x} \in \mathbb{F}_q^m} \left| \sum_{\mathbf{d} \in C^\perp} \chi_{\mathbf{v}}(\mathbf{x}-\mathbf{d}) \widetilde{P}(\mathbf{x} - \mathbf{d}) \right|^2.
    \end{align*}
    Expanding this out and using the complex conjugate \cref{eq:character_cc} and homomorphic \cref{eq:character_homomorphic} identities for characters, we have
    \begin{align*}
        &= \sum_{\mathbf{x} \in \mathbb{F}_q^m} \left( \sum_{\mathbf{d}_1,\mathbf{d}_2 \in C^\perp} \chi_{\mathbf{v}}(\mathbf{x}-\mathbf{d}_1) \widetilde{P}(\mathbf{x} - \mathbf{d}_1) \chi_{\mathbf{v}}(\mathbf{d}_2-\mathbf{x}) \widetilde{P}^*(\mathbf{x} - \mathbf{d}_2) \right) \\
        &= \sum_{\mathbf{x} \in \mathbb{F}_q^m} \left( \sum_{\mathbf{d}_1,\mathbf{d}_2 \in C^\perp} \chi_{\mathbf{v}}(\mathbf{d}_2 - \mathbf{d}_1) \widetilde{P}(\mathbf{x} - \mathbf{d}_1) \widetilde{P}^*(\mathbf{x} - \mathbf{d}_2) \right).
    \end{align*}
    By character orthogonality \cref{eq:character_orthogonality} it follows that
    \begin{align*}
        \mathbb{E}_{\mathbf{v}} \left[ \frac{1}{\mathcal{N}_{\mathrm{target}}^2} \right] &= \sum_{\mathbf{x} \in \mathbb{F}_q^m} \sum_{\mathbf{d} \in C^\perp} |\widetilde{P}(\mathbf{x} - \mathbf{d})|^2 \\
        &= |C^\perp|.
    \end{align*}
\end{proof}

\begin{lemma}\label{lem:innerproduct}
    If $U_{\mathrm{dec}}$ fails with probability $\epsilon$ then, for uniformly random $\mathbf{v} \in \mathbb{F}_q^m$,
    \begin{equation}
        \mathbb{E}_{\mathbf{v}} \left[ | \langle \Phi_{\mathrm{actual}}| \Phi_{\mathrm{target}} \rangle |^2 \right] \geq 1 - \epsilon.
    \end{equation}
\end{lemma}
\begin{proof}
    By \cref{lem:expectedinnerproduct},
    \begin{align*}
        |C^\perp|^2 (1-\epsilon)^2 &= \left( \mathbb{E}_\mathbf{v} \left[ \frac{\langle \Phi_{\mathrm{actual}} | \Phi_{\mathrm{target}} \rangle}{\mathcal{N}_{\mathrm{dec}} \mathcal{N}_{\mathrm{target}}}\right]\right)^2 \\
        &= \left| \mathbb{E}_\mathbf{v} \left[ \frac{\langle \Phi_{\mathrm{actual}} | \Phi_{\mathrm{target}} \rangle}{\mathcal{N}_{\mathrm{dec}} \mathcal{N}_{\mathrm{target}}}\right]\right|^2.
    \end{align*}
    Recalling that $\mathcal{N}_{\mathrm{target}}$ is $\mathbf{v}$-dependent and applying Cauchy-Schwarz yields
    \[
    \leq \mathbb{E}_{\mathbf{v}} \left[ |\langle \Phi_{\mathrm{actual}} | \Phi_{\mathrm{target}} \rangle|^2 \right] \cdot \mathbb{E}_{\mathbf{v}} \left[ \frac{1}{\mathcal{N}_{\mathrm{dec}}^2 \mathcal{N}_{\mathrm{target}}^2}\right].
    \]
    By \cref{lem:Ndec} this is equal to
    \[
        = \mathbb{E}_{\mathbf{v}} \left[ |\langle \Phi_{\mathrm{actual}} | \Phi_{\mathrm{target}} \rangle|^2 \right] \cdot \mathbb{E}_{\mathbf{v}} \left[ \frac{|C^\perp| (1-\epsilon)}{\mathcal{N}_{\mathrm{target}}^2}\right].
    \]
    By \cref{lem:Ntarg} this is equal to
    \[
        = \mathbb{E}_{\mathbf{v}} \left[ |\langle \Phi_{\mathrm{actual}} | \Phi_{\mathrm{target}} \rangle|^2 \right] \cdot |C^\perp|^2 (1-\epsilon).
    \]
    Thus we arrive at
    \[
        |C^\perp|^2 (1-\epsilon)^2 \leq \mathbb{E}_{\mathbf{v}} \left[ |\langle \Phi_{\mathrm{actual}} | \Phi_{\mathrm{target}} \rangle|^2 \right] \cdot |C^\perp|^2 (1-\epsilon),
    \]
    which implies the conclusion.
\end{proof}
\begin{lemma}\label{lem:diagH}
    Let $\hobj \in \mathbb{C}^{q^m \times q^m}$ be a diagonal matrix. Then we have:
    \begin{equation}
        \mathbb{E}_{\mathbf{v}} \left[ \frac{\langle \Phi_\mathrm{target}(\mathbf{v}) | \mathrm{QFT} X_q^{\mathbf{v}} \hobj X_q^{-\mathbf{v}} \mathrm{QFT}^\dag \otimes \id_\mathrm{anc} | \Phi_\mathrm{target}(\mathbf{v}) \rangle}{|C^\perp| \mathcal{N}_\mathrm{target}^2}\right] = \braket{\mathcal{S} | \hobj | \mathcal{S}},
    \end{equation}
    where $\ket{\mathcal{S}} = \sum_{\mathbf{x} \in \mathbb{F}_q^m} P(\mathbf{x}) \ket{\mathbf{x}}$.
\end{lemma}
\begin{proof}
    Define
    \[
        \ket{\mathcal{S}(\mathbf{v})} = \sum_{\mathbf{x} \in \mathbb{F}_q^m} P(\mathbf{x}-\mathbf{v}) \ket{\mathbf{x}},
    \]
    so that in particular we have $\ket{\mathcal{S}(\mathbf{0})} = \ket{\mathcal{S}}$.
    From \cref{eq:tildephi}, \cref{def:code_states}, and \cref{def:pointwise_product}, one has
    \[
        \ket{\widetilde{\Phi}_\mathrm{target}} = \sqrt{q^m |C^\perp|} \ \mathcal{N}_{\mathrm{target}} \ket{C} \odot \ket{\mathcal{S}(\mathbf{v})}.
    \]
    Therefore,
    \begin{align*}
        \frac{\bra{\Phi_{\mathrm{target}}} \mathrm{QFT} X_q^{\mathbf{v}} \hobj X_q^{-\mathbf{v}} \mathrm{QFT}^\dag \otimes \id_\mathrm{anc} \ket{\Phi_{\mathrm{target}}}}{q^m |C^\perp| \mathcal{N}_{\mathrm{target}}^2} &= \braket{C \odot \mathcal{S}(\mathbf{v}) | X_q^{\mathbf{v}} \hobj X_q^{-\mathbf{v}} | C \odot \mathcal{S}(\mathbf{v})} \\
        &= \braket{(C-\mathbf{v}) \odot \mathcal{S} | \hobj | (C-\mathbf{v}) \odot \mathcal{S}} \\
        &= \frac{1}{|C|} \sum_{\mathbf{x} \in \mathbb{F}_q^m} \braket{\mathbf{x} | \hobj | \mathbf{x}} |P(\mathbf{x})|^2 \mathbf{1}[\mathbf{x} \in C - \mathbf{v}].
    \end{align*}
    Consequently,
    \begin{align*}
        \mathbb{E}_{\mathbf{v}} \left[ \frac{\bra{\Phi_{\mathrm{target}}} \mathrm{QFT} X_q^{\mathbf{v}} \hobj X_q^{-\mathbf{v}} \mathrm{QFT}^\dag \otimes \id_\mathrm{anc}\ket{\Phi_{\mathrm{target}}}}{q^m |C^\perp| \mathcal{N}_{\mathrm{target}}^2} \right] &= q^{-m} \sum_{\mathbf{x} \in \mathbb{F}_q^m} \braket{\mathbf{x} | \hobj |\mathbf{x}} |P(\mathbf{x})|^2 \\
        &= q^{-m} \braket{\mathcal{S} | \hobj | \mathcal{S}},
    \end{align*}
    from which the conclusion follows.
\end{proof}

\subsection{Applying the Lemmas}
\label{sec:applying_lemmas}

In this subsection we use the lemmas established in the previous subsection to prove \cref{thm:expectation_bound}. We also prove that the distribution over bit strings obtained by measuring the final output $\ket{\Phi_{\mathrm{actual}}}$ from Regev's reduction in the computational basis is close in total variation distance to the distribution produced by measuring the ideal state $\ket{\Phi_{\mathrm{target}}}$, provided the failure rate $\epsilon$ of $U_{\mathrm{dec}}$ is small. More precisely, we have the following.

\begin{definition}\label{def:outputdistributions}
    Let $\mathcal{D}_{\mathrm{actual}}(\mathbf{v})$ (resp. $\mathcal{D}_{\mathrm{target}}(\mathbf{v})$) be the distribution over $\mathbb{F}_q^m$ obtained by measuring the first register of $(\mathrm{QFT}^\dag \otimes \id_\mathrm{anc}) \ket{\Phi_\mathrm{actual}}$ (resp. $(\mathrm{QFT}^\dag \otimes \id_\mathrm{anc}) \ket{\Phi_{\mathrm{target}}}$) in the standard basis. Let $\mathcal{D}_{\mathrm{algo}}(\mathbf{v})$ denote the final output distribution of Regev's reduction, which is $\mathcal{D}_\mathrm{actual}(\mathbf{v})$ conditioned on $\mathbf{x}$ being a codeword in $C$. Thus $\mathcal{D}_{\mathrm{algo}}, \mathcal{D}_{\mathrm{target}}$ are supported on $C$, but $\mathcal{D}_{\mathrm{actual}}$ need not be.
\end{definition}

Generalizing~\cite{chailloux2024quantum}, we begin by showing as a consequence of \cref{lem:innerproduct} that $\mathcal{D}_{\mathrm{actual}}, \mathcal{D}_{\mathrm{target}}, \mathcal{D}_{\mathrm{algo}}$ (as defined in \cref{def:outputdistributions}) are all statistically close:
\begin{theorem}\label{thm:distributionalcloseness}
    If $U_{\mathrm{dec}}$ fails with probability $\epsilon$, then for uniformly random $\mathbf{v} \in \mathbb{F}_q^m$,
    \begin{equation}
        \mathbb{E}_{\mathbf{v}} \left[ \Delta[\ket{\Phi_{\mathrm{actual}}},\ket{\Phi_{\mathrm{target}}}] \right] \leq \sqrt{\epsilon},
    \end{equation}
    where $\Delta$ denotes trace distance, and therefore
    \begin{equation}
        \mathbb{E}_{\mathbf{v}} \left[\Delta(\mathcal{D}_\mathrm{actual}(\mathbf{v}), \mathcal{D}_\mathrm{target}(\mathbf{v}))\right] \leq \sqrt{\epsilon},
    \end{equation}
    where $\Delta$ denotes the total variation distance.
\end{theorem}
\begin{proof}
    The trace distance between a pair of pure states $\ket{\Phi_{\mathrm{target}}}$ and $\ket{\Phi_{\mathrm{actual}}}$ is given by
    \[
        \Delta[\ket{\Phi_{\mathrm{actual}}},\ket{\Phi_{\mathrm{target}}}] =  \sqrt{1-|\langle \Phi_{\mathrm{target}}| \Phi_{\mathrm{actual}}\rangle|^2}.
    \]
    Since $\sqrt{1-x}$ is concave, we can apply Jensen's inequality to obtain
    \[
        \mathbb{E}_{\mathbf{v}} \left[ \Delta[\ket{\Phi_{\mathrm{actual}}},\ket{\Phi_{\mathrm{target}}}] \right] \leq \sqrt{1-\mathbb{E}_{\mathbf{v}} |\langle \Phi_{\mathrm{target}}| \Phi_{\mathrm{actual}}\rangle|^2}.
    \]
    By \cref{lem:innerproduct}, $\mathbb{E}_{\mathbf{v}} |\langle \Phi_{\mathrm{target}}| \Phi_{\mathrm{actual}}\rangle|^2 \geq 1 - \epsilon$, so 
    \[
        \mathbb{E}_{\mathbf{v}} \left[ \Delta[\ket{\Phi_{\mathrm{actual}}},\ket{\Phi_{\mathrm{target}}}] \right] \leq \sqrt{\epsilon}.
    \]
    If two quantum states have trace distance $\delta$ then for any choice of measurement, the distribution over outcomes produced by measuring one state will be within $\delta$ total variation distance of that produced by measuring the other state \cite{NC00}. Thus, we also conclude
    \[
        \mathbb{E}_{\mathbf{v}} \left[\Delta(\mathcal{D}_\mathrm{actual}(\mathbf{v}), \mathcal{D}_\mathrm{target}(\mathbf{v}))\right] \leq \sqrt{\epsilon}.
    \]
\end{proof}

\begin{theorem}\label{thm:projcodeword}
    If $U_{\mathrm{dec}}$ fails with probability $\epsilon$, then for uniformly random $\mathbf{v} \in \mathbb{F}_q^m$,
    \begin{equation}
        \mathbb{E}_{\mathbf{v}} \left[\Delta(\mathcal{D}_\mathrm{algo}(\mathbf{v}), \mathcal{D}_\mathrm{actual}(\mathbf{v}))\right] \leq \sqrt[4]{\epsilon},
    \end{equation}
    where $\Delta$ denotes the total variation distance.
\end{theorem}
\begin{proof}
    Define $\Pi_C$ to be a projector onto codewords in $C$. Then $\mathcal{D}_{\mathrm{algo}}(\mathbf{v})$ is obtained by measuring the first register of the state
    \[
        \ket{\Psi(\mathbf{v})} := \frac{\left(\Pi_C \mathrm{QFT}^\dag  \otimes \id_\mathrm{anc}\right)\ket{\Phi_{\mathrm{actual}}}}{\sqrt{\bra{\Phi_{\mathrm{actual}}} \mathrm{QFT} \Pi_C \mathrm{QFT}^\dag \otimes \id_\mathrm{anc} \ket{\Phi_{\mathrm{actual}}}}}
    \]
    in the computational basis. Let $\ket{\widetilde{\Phi}_{\mathrm{actual}}} = (\mathrm{QFT}^\dag \otimes \id_\mathrm{anc}) \ket{\Phi_{\mathrm{actual}}}$. Then, 
    \begin{align*}
        |\langle \Psi(\mathbf{v}) | \widetilde{\Phi}_{\mathrm{actual}} \rangle|^2 &= \bra{\Phi_{\mathrm{actual}}} \mathrm{QFT} \Pi_C \mathrm{QFT}^\dag \otimes \id_\mathrm{anc} \ket{\Phi_{\mathrm{actual}}} \\
        &\geq \bra{\Phi_{\mathrm{target}}} \mathrm{QFT} \Pi_C \mathrm{QFT}^\dag \otimes \id_\mathrm{anc} \ket{\Phi_{\mathrm{target}}} - \Delta(\ket{\Phi_{\mathrm{target}}},\ket{\Phi_{\mathrm{actual}}}) \\
        &= 1 - \Delta(\ket{\Phi_{\mathrm{target}}},\ket{\Phi_{\mathrm{actual}}})
    \end{align*}
    where $\Delta$ here indicates the trace distance. By \cref{thm:distributionalcloseness}, we thus have
    \[
        \mathbb{E}_{\mathbf{v}} \left[ | \langle \Psi(\mathbf{v})| \widetilde{\Phi}_{\mathrm{actual}} \rangle |^2 \right] \geq 1-\sqrt{\epsilon}.
    \]
    Because $\sqrt{1-x}$ is a concave function on the interval $[0,1]$ we can apply Jensen's inequality to obtain
    \[
        \mathbb{E}_{\mathbf{v}} \left[ \sqrt{1-| \langle \Psi(\mathbf{v})| \widetilde{\Phi}_{\mathrm{actual}} \rangle |^2 }\right] \leq \sqrt[4]{\epsilon}.
    \]
    The left-hand side is recognizable as the expected trace distance between $\ket{\Psi(\mathbf{v})}$ and $\ket{\widetilde{\Phi}_{\mathrm{actual}}}$, which upper bounds the total variation distance between the distribution over bit strings obtained by measuring these states. In other words
    \[
        \mathbb{E}_{\mathbf{v}} [\Delta(\mathcal{D}_{\mathrm{algo}},\mathcal{D}_{\mathrm{actual}})]\leq \mathbb{E}_{\mathbf{v}} \left[ \sqrt{1-| \langle \Psi(\mathbf{v})| \widetilde{\Phi}_{\mathrm{actual}} \rangle |^2 }\right],
    \]
 which, when combined with the prior equation, completes the proof.
\end{proof}
We next move on to theorems regarding the expected number of constraints satisfied (or general score attained) by the symbol string produced by Regev's reduction.

\begin{theorem}\label{thm:weirdcauchy}
    Let $\hobj'(\mathbf{v}) \in \mathbb{C}^{q^m \times q^m}$ be a Hermitian matrix (that could depend on $\mathbf{v}$) with spectral norm $\|\hobj'(\mathbf{v})\|_\infty \leq H_\mathrm{max}$ (for all $\mathbf{v}$). Then, if $U_{\mathrm{dec}}$ fails with probability $\epsilon$, we have for uniformly random $\mathbf{v} \in \mathbb{F}_q^m$:
    \begin{equation}
        \mathbb{E}_{\mathbf{v}} \left[\left|\braket{\Phi_\mathrm{actual}(\mathbf{v}) | \hobj'(\mathbf{v}) \otimes \id_\mathrm{anc} | \Phi_\mathrm{actual}(\mathbf{v})} - \frac{ \langle \Phi_\mathrm{target}(\mathbf{v}) | \hobj'(\mathbf{v}) \otimes \id_\mathrm{anc} | \Phi_\mathrm{target}(\mathbf{v}) \rangle}{\mathcal{N}_\mathrm{target}(\mathbf{v})^2 \cdot |C^\perp|} \right| \right] \leq 2 H_\mathrm{max} \cdot \sqrt{\epsilon}.
    \end{equation}
\end{theorem}
\begin{proof}
    For brevity, we use the subscripts $a, t$ to abbreviate ``actual'' and ``target''. Define the unnormalized state
    \[
        \ket{\Psi_t(\mathbf{v})} := \frac{1}{\mathcal{N}_t(\mathbf{v}) \cdot \sqrt{|C^\perp|}} \cdot \ket{\Phi_t(\mathbf{v})}.
    \]
    We have both of the following:
    \begin{align}
        \mathbb{E}_{\mathbf{v}} \left[ \langle \Psi_t(\mathbf{v}) | \Psi_t(\mathbf{v}) \rangle \right] &= 1 \quad \textrm{ (by Lemma~\ref{lem:Ntarg})} \label{eq:targ}\\
        \mathbb{E}_{\mathbf{v}}\left[\braket{\Phi_a(\mathbf{v}) | \Psi_t(\mathbf{v})}\right] &= \sqrt{1-\epsilon} \textrm{ (by Lemmas~\ref{lem:expectedinnerproduct} and~\ref{lem:Ndec})}.\label{eq:ip}
    \end{align}
    We now compute:
    \begin{align*}
        &2\left(\braket{\Phi_a(\mathbf{v}) | \hobj'(\mathbf{v}) \otimes \id_\mathrm{anc} | \Phi_a(\mathbf{v})} - \braket{\Psi_t(\mathbf{v}) | \hobj'(\mathbf{v})\otimes \id_\mathrm{anc} | \Psi_t(\mathbf{v})}\right) \\
        =\text{ }&\left(\bra{\Phi_a(\mathbf{v})} - \bra{\Psi_t(\mathbf{v})}\right) (\hobj'(\mathbf{v}) \otimes \id_\mathrm{anc}) \left(\ket{\Phi_a(\mathbf{v})} + \ket{\Psi_t(\mathbf{v})}\right) \\
        +\text{ }&\left(\bra{\Phi_a(\mathbf{v})} + \bra{\Psi_t(\mathbf{v})}\right) (\hobj'(\mathbf{v}) \otimes \id_\mathrm{anc}) \left(\ket{\Phi_a(\mathbf{v})} - \ket{\Psi_t(\mathbf{v})}\right),
    \end{align*}
    which implies:
    \begin{align*}
        &\left|\braket{\Phi_a(\mathbf{v}) | \hobj'(\mathbf{v}) \otimes \id_\mathrm{anc} | \Phi_a(\mathbf{v})} - \braket{\Psi_t(\mathbf{v}) | \hobj'(\mathbf{v}) \otimes \id_\mathrm{anc} | \Psi_t(\mathbf{v})}\right| \\
        \leq\text{ }&\|\ket{\Phi_a(\mathbf{v})} - \ket{\Psi_t(\mathbf{v})}\| \cdot H_\mathrm{max} \cdot \|\ket{\Phi_a(\mathbf{v})} + \ket{\Psi_t(\mathbf{v})}\|.
    \end{align*}
    By Cauchy-Schwarz, it follows that:
    \begin{align}
        &\mathbb{E}_{\mathbf{v}} \left|\braket{\Phi_a(\mathbf{v}) | \hobj'(\mathbf{v}) \otimes \id_\mathrm{anc} | \Phi_a(\mathbf{v})} - \braket{\Psi_t(\mathbf{v}) | \hobj'(\mathbf{v}) \otimes \id_\mathrm{anc} | \Psi_t(\mathbf{v})}\right| \\
        \leq\text{ }&H_\mathrm{max} \cdot \mathbb{E}_{\mathbf{v}} \left[\|\ket{\Phi_a(\mathbf{v})} - \ket{\Psi_t(\mathbf{v})}\| \cdot \|\ket{\Phi_a(\mathbf{v})} + \ket{\Psi_t(\mathbf{v})}\|\right] \\
        \leq\text{ }& H_\mathrm{max} \cdot \sqrt{\mathbb{E}_{\mathbf{v}} \left[\|\ket{\Phi_a(\mathbf{v})} - \ket{\Psi_t(\mathbf{v})}\|^2\right]} \cdot \sqrt{\mathbb{E}_{\mathbf{v}} \left[\|\ket{\Phi_a(\mathbf{v})} + \ket{\Psi_t(\mathbf{v})}\|^2\right]}.\label{eq:distbound}
    \end{align}
    From Equations~\eqref{eq:targ} and~\eqref{eq:ip} we have:
    \begin{align}
        \mathbb{E}_{\mathbf{v}} \left[\|\ket{\Phi_a(\mathbf{v})} - \ket{\Psi_t(\mathbf{v})}\|^2\right] &= 2 - 2\sqrt{1-\epsilon} \\
        \mathbb{E}_{\mathbf{v}}\left[\|\ket{\Phi_a(\mathbf{v})} + \ket{\Psi_t(\mathbf{v})}\|^2\right] &= 2 + 2\sqrt{1-\epsilon}.
    \end{align}
    Combining these with Equation~\eqref{eq:distbound} implies the conclusion.
\end{proof}

Lastly, we prove \cref{thm:expectation_bound}, which we restate here for convenience.

\mainErrorBound*

\begin{proof}
    Our algorithm will sample from $\mathcal{D}_{\mathrm{algo}}$. To analyze this, we apply \cref{thm:weirdcauchy} with $\hobj'(\mathbf{v}) := \mathrm{QFT} \Pi_C X_q^\mathbf{v} \hobj X_q^{-\mathbf{v}} \Pi_C \mathrm{QFT}^\dag$, where $\Pi_C$ is a projector onto strings in $C$. This has spectral norm $\leq 1$. By definition, we know that $(\mathrm{QFT}^\dag \otimes \id_\mathrm{anc}) \ket{\Phi_\mathrm{target}(\mathbf{v})}$ is entirely supported on strings in $C$. Moreover, we have:
    \begin{align}
        &\mathbb{E}_{\mathbf{v}} \left[\frac{\braket{\Phi_\mathrm{target}(\mathbf{v}) | \hobj'(\mathbf{v}) \otimes \id_\mathrm{anc} | \Phi_\mathrm{target}(\mathbf{v})}}{\mathcal{N}_\mathrm{target}(\mathbf{v})^2 \cdot |C^\perp|}\right] \\
        =\text{ }&\mathbb{E}_{\mathbf{v}} \left[\frac{\braket{\Phi_\mathrm{target}(\mathbf{v}) | \mathrm{QFT} \Pi_C X_q^\mathbf{v} \hobj X_q^{-\mathbf{v}} \Pi_C \mathrm{QFT}^\dag \otimes \id_\mathrm{anc} | \Phi_\mathrm{target}(\mathbf{v})}}{\mathcal{N}_\mathrm{target}(\mathbf{v})^2 \cdot |C^\perp|}\right] \\
        =\text{ }&\mathbb{E}_{\mathbf{v}} \left[\frac{\braket{\Phi_\mathrm{target}(\mathbf{v}) | \mathrm{QFT} X_q^\mathbf{v} \hobj X_q^{-\mathbf{v}} \mathrm{QFT}^\dag \otimes \id_\mathrm{anc} | \Phi_\mathrm{target}(\mathbf{v})}}{\mathcal{N}_\mathrm{target}(\mathbf{v})^2 \cdot |C^\perp|}\right] \\
        =\text{ }&\braket{\mathcal{S} | \hobj | \mathcal{S}} \text{ (Lemma~\ref{lem:diagH})}.
    \end{align}
    It follows by \cref{thm:weirdcauchy} that:
    \begin{align}\label{eq:cauchy1}
        \mathbb{E}_{\mathbf{v}} \left[\braket{\Phi_\mathrm{actual}(\mathbf{v}) | \hobj'(\mathbf{v})\otimes \id_\mathrm{anc} | \Phi_\mathrm{actual}(\mathbf{v})}\right] &\geq \braket{\mathcal{S} | \hobj | \mathcal{S}} - 2\sqrt{\epsilon}.
    \end{align}
    Now recall that $\mathcal{D}_{\mathrm{algo}}(\mathbf{v})$ is obtained by measuring the first register of the state
    \begin{equation}
        \ket{\Psi(\mathbf{v})} := \frac{1}{\sqrt{\braket{\Phi_\mathrm{actual}(\mathbf{v}) | \mathrm{QFT} \Pi_C \mathrm{QFT}^\dag \otimes \id_\mathrm{anc} | \Phi_\mathrm{actual}(\mathbf{v})}}}\left(\Pi_C \mathrm{QFT}^\dag\otimes \id_\mathrm{anc}\right) \ket{\Phi_\mathrm{actual}(\mathbf{v})}
    \end{equation}
    in the computational basis. 
    Equivalently, the mixed state $\rho(\mathbf{v})$ can be thought of as the result of computing $\ket{\Psi(\mathbf{v})}$ then discarding the trailing ancilla qubits.
    Bearing this in mind, we have:
    \begin{align}
        &\mathbb{E}_{\mathbf{v}} \left[\braket{\Phi_\mathrm{actual}(\mathbf{v}) | \hobj'(\mathbf{v}) \otimes \id_\mathrm{anc} | \Phi_\mathrm{actual}(\mathbf{v})}\right] \\
        =\text{ }&\mathbb{E}_{\mathbf{v}} \left[\braket{\Phi_\mathrm{actual}(\mathbf{v}) | \mathrm{QFT} \Pi_C X_q^{\mathbf{v}} \hobj X_q^{-\mathbf{v}} \Pi_C \mathrm{QFT}^\dag \otimes \id_\mathrm{anc} | \Phi_\mathrm{actual}(\mathbf{v})}\right] \\
        =\text{ }&\mathbb{E}_{\mathbf{v}} \left[\braket{\Phi_\mathrm{actual}(\mathbf{v}) | \mathrm{QFT} \Pi_C \mathrm{QFT}^\dag\otimes \id_\mathrm{anc} | \Phi_\mathrm{actual}(\mathbf{v})} \cdot \braket{\Psi(\mathbf{v}) | X_q^{\mathbf{v}} \hobj X_q^{-\mathbf{v}}\otimes \id_\mathrm{anc} |\Psi(\mathbf{v})}\right] \\
        \leq\text{ }&\mathbb{E}_{\mathbf{v}}\left[\braket{\Psi(\mathbf{v}) | X_q^{\mathbf{v}} \hobj X_q^{-\mathbf{v}} \otimes \id_\mathrm{anc}|\Psi(\mathbf{v})}\right] \\
        =\text{ }&\mathbb{E}_{\mathbf{v}}\left[\tr\left[X_q^{\mathbf{v}} \hobj X_q^{-\mathbf{v}} \rho(\mathbf{v})\right]\right].
    \end{align}
    Combining this with \cref{eq:cauchy1} implies the conclusion.
\end{proof}

\section{Turbo Prange}
\label{sec:turboprange}

In a max-XORSAT problem we have $m$ $\mathbb{F}_2$-linear equations and $n$ variables, where $m > n$. In Prange's algorithm we choose $\nu$ of the $m$ linear equations with $\nu$ as large as possible such that the resulting system remains solvable. Then we solve it using Gaussian elimination. Generically, the remaining $m-\nu$ equations that were not solved for will each be satisfied with probability half. For generic max-XORSAT instances, in which the equations are random, one expects the linear system to remain solvable up until $\nu = n - O(1)$ in the limit $n \to \infty$. Thus Prange's algorithm asymptotically yields solutions satisfying $n+(m-n)/2$ of the $m$ constraints.

The basic idea of our Turbo Prange heuristic is to first run Prange's algorithm to provide an initial solution, and then to further improve the solution using greedy bit flip updates. To maximize the success of the greedy updates, we choose $\nu$ equations to solve for in a way that is informed by the structure of the Tanner graph of the max-XORSAT instance.

Consider the Tanner graph arising from Gallager's ensemble shown in \cref{fig:gallager_tanner}. In the context of max-XORSAT the top row of vertices (labeled check nodes due to their role in $C^\perp$) represent the $n$ variables, and the bottom row of vertices (labeled bit nodes due to their role in $C^\perp$) represent the $m$ constraint equations. For Tanner graphs arising from Gallager's ensemble we can always choose a subset of the variable nodes that yields a partition of the equation nodes, as illustrated in blue.

Now, suppose we flip one of the blue variable nodes that defines the partition. If this variable appears in $D$ equations, $s$ of which are satisfied, then flipping it will make $D-s$ of these equations satisfied. In our greedy update phase, we can go through each of the blue variable nodes and consider such a change, making the change if it improves the objective value, \textit{i.e.} if $D-s > s$. Since the sets of equations that these blue variable nodes are connected to do not intersect, none of these bit flip moves will undo the gains from any previous move.

In the initial Gaussian elimination phase we have free choice of which $\nu$ equations to solve for. If we choose these randomly, then each of the sets of $D$ equations that we subsequently consider flipping will have on average $s=\frac{D}{m}(\nu+(m-\nu)/2)$, which is more than half of the $D$ equations. Thus, by the statistics of the binomial distribution, this makes it unlikely that $D-s > s$.

Instead, it is statistically better to pack all of the $\nu$ equations that we solve into $\nu/D$ of the elements of the partition. Then, in the greedy phase, we have $\nu/D$ variables where we are guaranteed that flipping them will definitely not yield improvement (since all equations containing them are satisfied), and $(m-\nu)/D$ variables such that flipping them has roughly a fifty percent chance of yielding improvement.

Packing the $\nu$ equations that we solve for into $\nu/D$ of the elements of the partition comes at a cost, however. There are $(m-\nu)/D$ elements of the partition that we exclude, and each of these has a variable that defines it. By definition, the sets in the partition do not overlap. Therefore, these defining variables will not be included in any of the chosen $\nu$ equations. The number of variables appearing in the chosen equations will be at most $n - (m-\nu)/D$. Additionally, some variables are expected to fail to appear in these equations by random chance, just as in our analysis of the erasure errors in \cref{sec:FGUM_analysis}. In fact, one recognizes that the maximum value of $\nu/m$ such that the linear system of equations is solvable with high probability is exactly the quantity $e_{\max}$ analyzed in \cref{sec:FGUM_analysis}.

We can now predict the fraction of constraints satisfied asymptotically by the Turbo Prange heuristic. For each of the unsolved elements of the partition the expected number of equations satisfied after the greedy update step is
\begin{equation}
    \label{eq:sigmad}
    \sigma_D = \frac{1}{2^D} \sum_{s=0}^D \binom{D}{s} \times \max \{s,D-s\}.
\end{equation}
Our partition contains $(m-e_{\max}m)/D$ such sets of equations plus $e_{\max}m/D$ sets of equations that are fully solved by Gaussian elimination. Thus, the expected fraction of equations satisfied will be
\begin{equation}
    \label{eq:stp}
    \frac{s_{\mathrm{TP}}}{m} = e_{\max} + \frac{\sigma_D}{D}\left( 1 - e_{\max} \right).
\end{equation}
Using \cref{eq:hatid} and \cref{eq:sigmad}, one can verify that
\begin{equation}
    \label{eq:sigma_vs_ihat}
    \sigma_D = D + (2^{1-D}-1) \hat{I}^*_D.
\end{equation}
Substituting \cref{eq:sigma_vs_ihat} into \cref{eq:stp} and comparing with \cref{eq:simple_satfrac} shows that 
\begin{equation}
\frac{s_{\mathrm{TP}}}{m} = \frac{\langle s \rangle}{m}.
\end{equation}
In other words, Turbo Prange matches the performance of Regev+FGUM for all $k,D$. Since $e_{\max}$ cancels out of this equality, this conclusion holds exactly even though our equation \cref{eq:emax} for the value of $e_{\max}$ as a function of $k$ and $D$ is only approximate.

\section{Comparison with QAOA}
\label{sec:QAOA}
The Quantum Approximate Optimization Algorithm (QAOA)~\cite{farhi2014quantum} is a quantum algorithm for combinatorial optimization. For a $D$-regular max-$k$-XORSAT instance specified by a $(k,D)$-regular Tanner graph with $n$~variable nodes (set~$V$) and $m$~constraint nodes, the depth-$p$ QAOA state is
\begin{equation}\label{eq:qaoa_state}
  \ket{\boldsymbol{\gamma},\boldsymbol{\beta}}
  = \prod_{\ell=1}^{p} U_B(\beta_\ell)\,U_C(\gamma_\ell)
  \;\ket{+}^{\otimes n}\,,
\end{equation}
where, denoting by $e_1,\ldots,e_k$ the $k$ variable-node neighbors of constraint~$e$,
$U_C(\gamma) = \prod_{e=1}^{m} e^{-i\gamma\, Z_{e_1}\cdots Z_{e_k}}$
is the phase operator and
$U_B(\beta) = \prod_{v\in V} e^{-i\beta\, X_v}$
is the mixer.
The $2p$ parameters $(\boldsymbol\gamma,\boldsymbol\beta)$ are chosen to maximize the expected fraction of satisfied constraints
\begin{equation}\label{eq:qaoa_observable}
  \frac{s_{\mathrm{QAOA}}}{m}
  = \frac{1}{m}\sum_{e=1}^{m}
    \frac{1+\bra{\boldsymbol{\gamma},\boldsymbol{\beta}}\,
    Z_{e_1}\cdots Z_{e_k}\,
    \ket{\boldsymbol{\gamma},\boldsymbol{\beta}}}{2}\,.
\end{equation}

\subsection{Finite $k$ and $D$}\label{sec:QAOA_finite_k_D}
We begin by evaluating the QAOA energy on $D$-regular max-$k$-XORSAT for finite $D$ and $k$ but in the $n\rightarrow\infty$ limit.  Due to the locality of QAOA~\cite{farhi2014quantum}, the expectation $\bra{\bm{\gamma}, \bm{\beta}}Z_{e_1} \ldots Z_{e_k} \ket{\bm{\gamma}, \bm{\beta}}$ for constant QAOA depth depends only on a constant-sized neighborhood of the constraint in the Tanner graph. We note that for Gallager's ensemble and large $n$, the fraction of constraint neighborhoods that are tree-like is $1 - O(n^{-(1-\epsilon)})$ for every $\epsilon > 0$, with high probability over the random choice of instance (see~\cref{app:trees}); we may set $\epsilon = 0.01$ for concreteness. Since the contribution of any per-constraint expectation is bounded by $1$, the QAOA energy density on a random Gallager instance equals the tree value up to $O(n^{-0.99})$ correction with high probability. It is therefore sufficient to assume a high-girth Tanner graph to compute the performance of QAOA on the Gallager's ensemble for large $n$.

Under this high-girth condition, the expectation value~\eqref{eq:qaoa_observable} reduces to contracting a tensor network on a single light-cone tree since all terms contribute equally. To perform this contraction, we extend the techniques introduced for $k=2$ in \cite{BF21, farhi2025lower} to arbitrary $k$. By regularity, the $D{-}1$ identical subtrees at each variable collapse into a single branch, and the $k{-}1$ identical children of each clause likewise collapse, giving a total cost of $O(p \cdot 4^p)$, independent of $D$, $k$, and the total number of qubits in the light-cone (which is exponential in $k$, $D$, and $p$). We use the resulting objective to optimize the QAOA parameters $(\bm{\gamma},\bm{\beta})$ for depths $p\leq16$. 

Results for all $(k,D)$ in \cref{tab:scores} are reported for $p = 16$. At this depth, QAOA provably outperforms Regev+FGUM for configurations $(k,D)\in\left\{(3,5), (3,6), (3,7), (3,8), (4,8)\right\}$ at sufficiently large $n$. For these configurations, simulated annealing surpasses both Regev+FGUM and QAOA. However, the QAOA performance grows monotonically with depth $p$. Although we stop at $p=16$ because our analysis methods become too computationally costly at larger $p$, running QAOA on a quantum computer would have cost only linear in $p$. Investigating whether QAOA surpasses competing quantum and classical algorithms for $D$-regular max-$k$-XORSAT at higher $p$ is an interesting area for future work.

\subsection{Large $D$ Asymptotics}
In \cite{BF21} it was shown that the fraction of constraints satisfied by depth-$p$ QAOA on $D$-regular max-$k$-XORSAT is
\begin{equation}
    \label{eq:QAOA}
    \frac{s_{\mathrm{QAOA}}}{m} = \frac{1}{2} + \nu_p^{[k]} \sqrt{\frac{k}{2(D-1)}} \pm O(1/D).
\end{equation}
Numerically computed values of $\nu_p^{[k]}$ for $k$ up to 6 are given at \cite{bensgithub}.

The result \cref{eq:QAOA} provides an insight into the behavior of QAOA in the limit of large $D$ with $k$ fixed, complementing the analysis for finite $D$ and $k$ in \cref{sec:QAOA_finite_k_D} and \cref{tab:scores}. This motivates an analysis of Regev+FGUM in the same limit. Recall that Regev+FGUM matches the performance of Turbo Prange. Since the formulas for the performance of Turbo Prange are slightly simpler, we will analyze their asymptotics.

Let $S$ be a random variable binomially distributed according to $S \sim \mathrm{Bin}(D,1/2)$. From \cref{eq:sigmad}, we see that
\begin{equation}
    \label{eq:sigmaD_largeD}
    \frac{\sigma_D}{D} = \frac{1}{2} + \frac{\mathbb{E}\left[|S-D/2|\right]}{D} = \frac{1}{2} + \frac{1}{\sqrt{2 \pi D}} + O(D^{-1}).
\end{equation}
Next, we show that for fixed $k \geq 3$, $e_{\max}$ scales like $1/D$ in the limit of large $D$. To do so, substitute $e_{\max} = x_k/{D}$ into \cref{eq:emax}. This yields
\begin{equation}
    x_k \simeq \frac{x_k}{D} + (k-1) \left( 1 - \left( 1-\frac{x_k}{D} \right)^D\right).
\end{equation}
Recall that
\begin{equation}
    \lim_{D \to \infty} \left( 1-\frac{x_k}{D} \right)^D = e^{-x_k}.
\end{equation}
Thus, 
\begin{equation}
    \label{eq:xk}
    x_k \simeq (k-1) \left( 1 - e^{-x_k} \right) \quad \textrm{for large $D$}.
\end{equation}
We see that the solution $x_k$ to this equation becomes independent of $D$ in the limit of large $D$. Hence
\begin{equation}
    e_{\max} \simeq \frac{x_k}{D} \quad \textrm{for large $D$},
\end{equation}
where $x_k$ is a constant that depends only on $k$ and can be found by numerically solving \cref{eq:xk}.

Substituting \cref{eq:xk} and \cref{eq:sigmaD_largeD} into \cref{eq:stp} and recalling that $s_{TP} = \langle s \rangle$, we see that the satisfaction fraction achieved by Regev+FGUM for fixed $k$ in the limit of large $D$ is
\begin{equation}
    \frac{ \langle s \rangle}{m} = \frac{1}{2} + \frac{1}{\sqrt{2 \pi D}} + O(D^{-1}).
\end{equation}
In \cref{tab:nu} we compare this with the satisfaction fraction achieved by 14-round QAOA in the limit of fixed $k$ and large $D$ and find that QAOA wins.

\begin{table}[ht!]
    \centering
    \renewcommand{\arraystretch}{1.4}
    \begin{tabular}{c|cccc}
    \hline
        $k$ & 3 & 4 & 5 & 6 \\
        \hline
        $\nu^k_{14} \sqrt{k/2}$ & 0.7865 & 0.8666 & 0.9243 & 0.9686 \\
        \hline
    \end{tabular}
    \caption{\label{tab:nu} Asymptotically, QAOA with $p$ rounds achieves a satisfaction fraction $\frac{1}{2} + \nu^k_p \sqrt{k/2} D^{-1/2} + O(1/D)$. The values of $\nu^k_p$ are computed up to $p=14$ at \cite{bensgithub}. Regev+FGUM achieves $\frac{1}{2} + \frac{1}{\sqrt{2 \pi}} D^{-1/2} + O(1/D)$. As shown above, $\nu^k_{14} \sqrt{k/2}$ exceeds $\frac{1}{\sqrt{2 \pi}} \simeq 0.3989$ for all values of $k$ considered in \cite{BF21}. Thus QAOA with 14 rounds outperforms Regev+FGUM in the limit of large $D$.}
\end{table}

\subsection{QAOA at higher depth}

The large-$D$ analysis shows that QAOA with $p=14$ rounds outperforms Regev+FGUM in the $D\to\infty$ limit at fixed $k$. The finite $(k,D)$ results in \cref{sec:QAOA_finite_k_D} show that at $p=16$ QAOA surpasses Regev+FGUM for $(k,D) \in \{(3,5),(3,6),(3,7),(3,8),(4,8) \}$, whereas Regev+FGUM surpasses $p=16$ QAOA for $(k,D) \in \{(3,4),(4,5),(4,6),(4,7),(5,6),(5,7),(5,8),(6,7),(6,8),(7,8) \}$. Since the performance of QAOA increases with $p$, it may be that QAOA at larger $p$ surpasses Regev+FGUM at additional $(k,D)$ pairs. Refuting or confirming this remains an interesting open question, as does assessing the relative performance of QAOA and Regev+FGUM when the ratio $k/D$ is kept fixed and both become large.

\section{Conclusions}
\label{sec:conclusions}

Locally-quantum decoding substantially broadens the toolbox for quantum decoders in DQI and related algorithms. The measurements we analyze are essentially the simplest local measurement operators one could use, fitting in conveniently with Gaussian elimination. This baseline already substantially outperforms classical belief propagation and brings us to the boundary of quantum advantage. However, parity equations are by no means the only useful information we can learn. Using more sophisticated coherent measurements or adaptive strategies, one might hope to cross the boundary and realize genuine quantum advantage on max-$k$-XORSAT. \\

\noindent
\textbf{Acknowledgments:} We thank Umesh Vazirani, Alexander Schmidhuber, Mark Zhandry, Eddie Farhi, and Vinod Vaikuntanathan for helpful conversations and feedback on the manuscript. SR was supported by the Defense Advanced Research Projects Agency (DARPA) under Contract No. HR0011-25-C-0300 and Amazon Research Awards, and partially supported by Jane Street. Any opinions, findings and conclusions or recommendations expressed in this material are those of the author(s) and do not necessarily reflect the views of the Defense Advanced Research Projects Agency (DARPA). QB and AC acknowledge funding from the French PEPR integrated projects EPIQ (ANR-22-PETQ-007), PQ-TLS (ANR-22-PETQ-008) and HQI (ANR-22-PNCQ-0002) all part of plan France 2030. This work was done in part while the authors were visiting the Simons Institute for the Theory of Computing and the Challenge Institute for Quantum Computation at UC Berkeley.\\

\noindent
\textbf{Disclaimer:} This paper was prepared for informational purposes with contributions from the Global Technology Applied Research center of JPMorgan Chase \& Co. This paper is not a product of the Research Department of JPMorgan Chase \& Co. or its affiliates. Neither JPMorgan Chase \& Co. nor any of its affiliates makes any explicit or implied representation or warranty and none of them accept any liability in connection with this paper, including, without limitation, with respect to the completeness, accuracy, or reliability of the information contained herein and the potential legal, compliance, tax, or accounting effects thereof. This document is not intended as investment research or investment advice, or as a recommendation, offer, or solicitation for the purchase or sale of any security, financial instrument, financial product or service, or to be used in any way for evaluating the merits of participating in any transaction.

\begin{appendices}
\crefalias{section}{appendix}

\section{Relation Between Regev's Reduction and DQI}
\label{app:DQI_vs_regev}

Although DQI was obtained in \cite{JSW25} from a different conceptual framework, it bears a close relation to Regev's reduction, which we here clarify. Aside from the conceptual framing, the primary difference between DQI and Regev's reduction is that, whereas Regev's reduction prepares the state
\[
    \ket{\psi_P} \propto \sum_{\mathbf{x} \in \mathbb{F}_2^n} P(B\mathbf{x} - \mathbf{v}) \ket{B \mathbf{x}},
\]
DQI prepares the state
\[
    \ket{\psi_{\mathrm{DQI}}} \propto \sum_{\mathbf{x} \in \mathbb{F}_2^n} P(B\mathbf{x} - \mathbf{v}) \ket{\mathbf{x}}.
\]
That is, DQI directly prepares a weighted superposition over the decision variables $\mathbf{x} \in \mathbb{F}_2^n$, whereas Regev's reduction prepares a weighted superposition over strings $B \mathbf{x} \in \mathbb{F}_2^m$ indicating which of the $m$ constraints are satisfied.

In the case that $B \in \mathbb{F}_2^{m \times n}$ has rank $n$, the sets of strings $\mathbf{x} \in \mathbb{F}_2^n$ and $\{B \mathbf{x} : \mathbf{x} \in \mathbb{F}_2^n \}$ are in bijective correspondence and furthermore the correspondence is efficiently implementable on classical computers using Gaussian elimination over $\mathbb{F}_2$. In this case the $m$-qubit state $\ket{\psi_P}$ and the $n$-qubit state $\ket{\psi_{\mathrm{DQI}}}$ are equally useful computationally, up to polynomial factors, although $\ket{\psi_{\mathrm{DQI}}}$ uses fewer qubits and requires no classical postprocessing.

In the language of codes, Regev's reduction prepares a superposition over codewords in the primal code $C = \{ B \mathbf{x} : \mathbf{x} \in \mathbb{F}_2^n \}$ which is weighted by the distance of these codewords from a target bit string $\mathbf{v}$. It reduces this to a problem of decoding codewords from the dual code $C^\perp = \{\mathbf{d} \in \mathbb{F}_2^m : B^T \mathbf{d} = \mathbf{0} \}$ which have been corrupted by bit flip errors. DQI produces a superposition of messages $\mathbf{x}$ weighted by the distance of their corresponding primal codewords $B \mathbf{x}$ from a target bit string $\mathbf{v}$. It reduces this to the problem of deducing errors $\mathbf{e} \in \mathbb{F}_2^m$ from the corresponding syndromes $B^T \mathbf{e} \in \mathbb{F}_2^n$ induced by the parity check matrix $B^T$ of the dual code $C^\perp$.

A second difference between DQI and Regev's reduction is that in DQI the function $P$ is chosen so that its Fourier transform $\widetilde{P}$ has support only out to Hamming radius $\ell$, and among all possible such functions $\widetilde{P}$ the optimal choice is found for maximizing the expected number of max-XORSAT constraints satisfied upon measuring the DQI output state. In some cases, such as for the OPI problem, $\ell$ can be chosen such that the decoding problem out to radius $\ell$ can be solved with $100\%$ success probability, which ensures that the quantum algorithm requires no postselection and achieves perfectly uniform sampling among solutions of given objective value. This choice of $\widetilde{P}$ necessitates some additional quantum circuitry to prepare a superposition of Dicke states in the initial step of the DQI algorithm \cite{JSW25, khattar2025verifiable}.

In Regev's reduction one instead traditionally uses a product distribution for $P$, such as 
\[
    P_\alpha(\mathbf{y}) := \prod_{j=1}^m \sqrt{1-\alpha}^{1-y_j} \sqrt{\alpha}^{y_j}.
\]
This has the advantage that the first state preparation step in the algorithm is trivial, as one needs only to prepare a tensor product state. However, the Hadamard transform
\[
    \widetilde{P}_\alpha(\mathbf{e}) := \prod_{j=1}^m \sqrt{1-\alpha^\perp}^{1-e_j} \sqrt{\alpha^\perp}^{e_j}
\]
has support on all bit strings. Consequently it is impossible to perfectly solve the resulting quantum decoding problem of recovering an unknown codeword $\mathbf{d}$ from the state $\sum_{\mathbf{e} \in \mathbb{F}_2^m} \widetilde{P}_\alpha(\mathbf{e}) \ket{\mathbf{d} + \mathbf{e}}$. At sufficiently small $\alpha^\perp$ the translations of this state by different codewords in $C^\perp$ have small overlap and $\mathbf{d}$ can be recovered with high probability, resulting in a good approximate solution to the quantum decoding problem. The error analysis is surprisingly subtle, however, as noted in Section 6.3.1 of \cite{chailloux2023quantum}. An alternative approach, as in \cite{debris2024quantum}, is to truncate $\widetilde{P}$ at some Hamming radius below half the distance of the code $C^\perp$, which ensures that perfect decodability is at least information-theoretically possible.

Remarkably, although the product distribution truncated to Hamming radius $\ell$ is not identical at any finite $\ell$ to the fully optimal choice $\widetilde{P}$ of radius-$\ell$ functions derived in \cite{JSW25}, the expected number of constraints satisfied using the truncated product distribution converges to the optimum asymptotically. Thus the analysis in \cite{JSW25} shows retrospectively that the choice in \cite{debris2024quantum} is asymptotically optimal.

For the purpose of solving optimization problems by reducing them to classical decoding problems the choice between DQI versus Regev's reduction is largely a matter of technical details that can affect quantitative resource requirements but do not affect what can or cannot be achieved in polynomial time. However, there are some contexts in which the distinction between these two approaches is more significant. In particular, DQI generalizes naturally to solve problems of preparing Gibbs states of Hamiltonians \cite{schmidhuber2025hamiltonian}. On the other hand, the quantum decoding problem produced by Regev's reduction, in which one is faced with recovering a codeword that has been subjected to a superposition of bit flip errors, appears more natural to work with than the quantum decoding problem produced by DQI, in which one is faced with recovering a superposition of error vectors from a superposition of syndromes. For this reason, in the present manuscript we work in the $m$-dimensional dual spaces of codewords and corrupted dual codewords, as is traditional in Regev's reduction, rather than in the $n$-dimensional dual spaces of messages and syndromes, as is done in DQI.

\section{Gallager's Ensemble is Locally Treelike}
\label{app:trees}

\begin{definition}\label{def:gallager-ensemble}
For positive integer $t$ let $[t] = \{1,2,\ldots,t\}$. Fix integers $D>k\ge 2$. For each $b\ge 1$, Gallager's ensemble $\mathcal{G}(k,D,b)$ is constructed as follows. Let
\[
A_0=[I_b\ I_b\ \cdots\ I_b]\in\{0,1\}^{\,b\times Db}
\]
be the horizontal concatenation of $D$ copies of the $b\times b$ identity matrix. For each $i\in [k]$, choose independently a uniformly random $Db\times Db$ permutation matrix $P_i$ and define
\[
M_i:=A_0P_i.
\]
Now stack these $k$ row blocks vertically to obtain the $kb\times Db$ matrix
\[
A=
\left[\begin{array}{c}
M_1\\ \hline
M_2\\ \hline
\vdots\\ \hline
M_k
\end{array}\right].
\]
The matrix $A$ defines a bipartite graph in the natural way, with rows corresponding to check vertices and columns corresponding to bit vertices, and with an edge between a check vertex and a bit vertex if and only if the corresponding entry of $A$ is $1$.
\end{definition}

By construction, in a graph sampled from $\mathcal{G}(k,D,b)$, every bit vertex has degree $k$ and every check vertex has degree $D$. Moreover, the $k$ layers of checks are independent and identically distributed, and within each layer the $b$ check vertices have disjoint neighborhoods that form a uniformly random partition of the common bit set into $b$ blocks of size $D$. Furthermore, each bit vertex is adjacent to exactly one check vertex from each layer.

More formally, we index the bit vertices as $V_b = [Db]$ and the check vertices as $C_b = \{(i,j): i\in [k], j\in [b]\}$, so that the check vertex $(i,j)$ corresponds to row $(i-1)b+j$ of $A$. For each $(i,j)\in C_b$, let $S_{(i,j)}\subseteq V_b$ denote the support of row $(i-1)b+j$. Equivalently, $S_{(i,j)}$ is the neighborhood of the check vertex $(i,j)$.  For each fixed $i\in [k]$, the checks
\[
(i,1),(i,2),\dots,(i,b)
\]
form the $i$th \emph{layer}. In that layer, the sets
\[
S_{(i,1)},\dots,S_{(i,b)}
\]
form a uniformly random ordered partition of the common bit set $V_b$ into $b$ disjoint blocks of size $D$. For each $i \in [k]$, let $\sigma_i$ denote the permutation on $[Db]$ underlying $P_i$, so that $(M_i)_{j,w} = 1$ iff $\sigma_i(w) \equiv j \pmod{b}$, or equivalently $w \in S_{(i,j)}$.

We are now ready to bound the expected number of short cycles.

\begin{proposition}[Bounding expected number of short cycles]
\label{prop:bounding_expected_number_short_cycles}
Let $N_{\ell}$ denote the number of $\ell$-cycles in the Tanner graph, i.e.\ the number of cycles $\left(u_1\,w_1\,\ldots\,u_{\ell}\,w_{\ell}\right)$ where $u_1, \ldots, u_{\ell}$ are distinct check vertices and $w_1, \ldots, w_{\ell}$ distinct bit vertices. Then for a graph sampled from $\mathcal{G}(k,D,b)$,
\begin{align}
    \mathbb{E}\left[N_{\ell}\right] & \leq \left(4kD\right)^{\ell} \qquad \textrm{for $b \geq 4\ell/D$}.
\end{align}
\begin{proof}
By linearity of expectation,
\begin{align}\label{eq:expected_l_cycles_final_expression}
    \mathbb{E}\left[N_{\ell}\right] & = \sum_{\substack{\textrm{candidate}\\\textrm{$\ell$-cycles}}}\mathrm{Pr}\left[\left(u_1\,w_1\,\ldots\,u_{\ell}\,w_{\ell}\right) \in G\right],
\end{align}
where the sum is over all potential $\ell$-cycles, i.e.\ all choices of distinct check vertices $u_1,\ldots,u_{\ell} \in C_b$ and distinct bit vertices $w_1,\ldots,w_{\ell} \in V_b$. We bound the two ingredients---the per-cycle probability and the number of candidate cycles---separately.

\medskip
\noindent\emph{Per-cycle probability.}
Write each check vertex as $u_r = (i_r, j_r) \in C_b$. The cycle $(u_1\,w_1\,\ldots\,u_{\ell}\,w_{\ell})$ belongs to $G$ iff for each $r \in [\ell]$ (indices mod $\ell$), check vertex $(i_r, j_r)$ is adjacent to both bit vertices $w_r$ and $w_{r-1}$, i.e.\ $w_r, w_{r-1} \in S_{(i_r,j_r)}$. Since $\sigma_{i_r}(w) \equiv j_r \pmod{b}$ iff $w \in S_{(i_r,j_r)}$, this requires:
\begin{align}\label{eq:residue_constraints}
    \sigma_{i_r}(w_r) \equiv j_r \pmod{b} \quad \textrm{and} \quad \sigma_{i_r}(w_{r-1}) \equiv j_r \pmod{b}.
\end{align}
Each residue class $\{j_r, j_r + b, \ldots, j_r + (D-1)b\}$ has $D$ elements. For each $r$, let $s_r, s'_r \in \{0, \ldots, D-1\}$ specify the exact images:
\begin{align}
    \sigma_{i_r}(w_r) = j_r + bs_r, \qquad \sigma_{i_r}(w_{r-1}) = j_r + bs'_r.
\end{align}
For any fixed choice of $(s_r, s'_r)_{r \in [\ell]}$, these are $2\ell$ point constraints on the uniform random permutations $\sigma_1,\ldots,\sigma_k$. Since the values of a uniform random permutation on $[Db]$ at $d$ distinct inputs are distributed as sampling without replacement (i.e., the probability of $d$ prescribed input--output pairs is $1/(Db)_d \leq 1/(Db - d)^d$), the probability that all constraints are simultaneously satisfied is at most:
\begin{align}
    \frac{1}{\left(Db - 2\ell\right)^{2\ell}}.
\end{align}
Taking a union bound over the $D^{2\ell}$ choices of $(s_r, s'_r)_{r \in [\ell]}$:
\begin{align}\label{eq:per_cycle_prob_bound}
    \mathrm{Pr}\left[(u_1\,w_1\,\ldots\,u_{\ell}\,w_{\ell}) \in G\right] & \leq \frac{D^{2\ell}}{\left(Db - 2\ell\right)^{2\ell}}.
\end{align}

\medskip
\noindent\emph{Counting candidate cycles.}
The number of candidate $\ell$-cycles is at most $|C_b|^{\ell}|V_b|^{\ell}/\ell = (kb)^{\ell}(Db)^{\ell}/\ell$: there are at most $(kb)^{\ell}$ ordered $\ell$-tuples of check vertices and $(Db)^{\ell}$ of bit vertices, and we divide by $\ell$ since each cycle is counted $\ell$ times under cyclic relabelling.

\medskip
\noindent\emph{Combining.}
Multiplying the per-cycle bound~\eqref{eq:per_cycle_prob_bound} by the cycle count:
\begin{align}
    \mathbb{E}\left[N_{\ell}\right] & \leq \frac{(kb)^{\ell}(Db)^{\ell}}{\ell} \cdot \frac{D^{2\ell}}{(Db - 2\ell)^{2\ell}}
    = \frac{k^{\ell}D^{\ell}}{\ell\left(1 - \frac{2\ell}{Db}\right)^{2\ell}}
    \leq \frac{\left(4kD\right)^{\ell}}{\ell}
    \leq \left(4kD\right)^{\ell} \qquad \textrm{for $b \geq 4\ell/D$}.
\end{align}
\end{proof}
\end{proposition}

Proposition~\ref{prop:bounding_expected_number_short_cycles} confirms the number of short cycles is constant in expectation. The following Corollary shows this implies almost-treelikeness of $G$ with high probability.

\begin{corollary}[Bounding fraction of non-treelike neighbourhoods with high probability]
\label{corollary:gallager_ensemble_treelikeness}
For any vertex $x$ of $G$, write $\mathcal{N}_{\ell}(x)$ for the subgraph of $G$ induced on all vertices within graph distance $2\ell$ of $x$. For $G \sim \mathcal{G}(k,D,b)$ with $b \geq 8\ell/D$,
\begin{align}
    \mathbb{E}\!\left[\left|\left\{x \in C_b \cup V_b\,:\,\mathcal{N}_{\ell}(x)\,\textrm{has a cycle}\right\}\right|\right] & \leq \left(4kD\right)^{3\ell + 2}.
\end{align}
In particular, for any constants $\ell$ and $\epsilon > 0$, $G$ is treelike around all but an $\mathcal{O}(b^{-(1-\epsilon)})$ fraction of its vertices with probability $1 - \mathcal{O}\!\left(b^{-\epsilon}\right)$.
\begin{proof}
Since check vertices have degree $D$ and bit vertices degree $k$, every vertex has at most $L:=(kD)^{\ell+1}$ vertices within distance~$2\ell$. Call $x$ \emph{bad} if it lies on a simple cycle of length $\ell'\le 2\ell$, and \emph{tainted} if $\mathcal{N}_{\ell}(x)$ contains a cycle.\footnote{Recall that our convention for cycle length in the Tanner graph counts the number of check--bit pairs: an $\ell'$-cycle traverses $\ell'$~check vertices and $\ell'$~bit vertices, hence $2\ell'$~edges.} Any cycle in $\mathcal{N}_{\ell}(x)$ contains a simple cycle of at most $4\ell$ edges, i.e.\ of length at most $2\ell$ in our convention (a non-tree edge in the BFS tree from $x$ closes a cycle of at most $4\ell+1$ edges, rounded down to $4\ell$ since $G$ is bipartite). So every tainted vertex is within distance $2\ell$ of a bad vertex, giving
\[
\#\textrm{tainted} \;\le\; L\cdot\#\textrm{bad} \;\le\; L\sum_{\ell'=1}^{2\ell}2\ell'\,N_{\ell'},
\]
where the last inequality uses that each simple $\ell'$-cycle contributes $2\ell'$ bad vertices. Taking expectations and applying Proposition~\ref{prop:bounding_expected_number_short_cycles} (valid for $b\ge 8\ell/D$),
\begin{align}
    \mathbb{E}\!\left[\#\textrm{tainted}\right]
    &\;\le\; L\sum_{\ell'=1}^{2\ell}2\ell'\left(4kD\right)^{\ell'}
    \;\le\; (kD)^{\ell+1}\cdot 8\ell\cdot\left(4kD\right)^{2\ell}
    \;=\; 8\ell\cdot 4^{2\ell}(kD)^{3\ell+1}\notag\\
    &\;\le\; 4^{3\ell+2}(kD)^{3\ell+1}
    \;\le\; \left(4kD\right)^{3\ell + 2},
\end{align}
where the second inequality bounds the sum by $2\ell$ times its largest term times a geometric factor ($\tfrac{r}{r-1}\le 2$ for $r=4kD\ge 16$), the fourth uses $\ell\le 4^{\ell}$ and $8\le 4^2$.
The high-probability statement follows by Markov's inequality.
\end{proof}
\end{corollary}

Corollary~\ref{corollary:gallager_ensemble_treelikeness} establishes that in the limit $b \to \infty$, a Tanner graph $G$ sampled from $\mathcal{G}(k,D,b)$ is treelike around all but a vanishing fraction of its vertices with high probability. The energy of QAOA for any hyperedge is further bounded by a one. It follows that for every $\epsilon > 0$, the energy density of QAOA on the $(k,D)$-Gallager ensemble equals its tree value, computed according to \cite{BF21}, up to $O(n^{-(1-\epsilon)})$ correction with high probability. Setting $\epsilon = 0.01$ for concreteness, the correction is $O(n^{-0.99})$. 

\end{appendices}

\bibliography{minrefs}

@inproceedings{chailloux2024quantum,
  title={Quantum advantage from soft decoders},
  author={Chailloux, Andr{\'e} and Tillich, Jean-Pierre},
  booktitle={Proceedings of the 57th Annual ACM Symposium on Theory of Computing},
  pages={738--749},
  year={2025}
}

@article{JSW25,
	title = {Optimization by decoded quantum interferometry},
	author = {Stephen P. Jordan and Noah Shutty and Mary Wootters and Adam Zalcman and Alexander Schmidhuber and Robbie King and Sergei V. Isakov and Tanuj Khattar and Ryan Babbush},
	date = {2025/10/01},
	doi = {10.1038/s41586-025-09527-5},
	isbn = {1476-4687},
	journal = {Nature},
	number = {8086},
	pages = {831--836},
	url = {https://doi.org/10.1038/s41586-025-09527-5},
	volume = {646},
	year = {2025},
  note={arXiv:2408.08292}
}

@inproceedings{chailloux2023quantum,
  title={The Quantum Decoding Problem},
  author={Chailloux, Andr{\'e} and Tillich, Jean-Pierre},
  booktitle={19th Conference on the Theory of Quantum Computation, Communication and Cryptography (TQC 2024)},
  pages={6--1},
  year={2024},
  organization={Schloss Dagstuhl--Leibniz-Zentrum f{\"u}r Informatik}
}

@inproceedings{clz21,
  title={Quantum algorithms for variants of average-case lattice problems via filtering},
  author={Yilei Chen and Qipeng Liu and Mark Zhandry},
  booktitle={Annual international conference on the theory and applications of cryptographic techniques},
  pages={372--401},
  year={2022},
  organization={Springer}
}

@article{kothari2025no,
  title={No exponential quantum speedup for $\mathrm{SIS}^{\infty}$ anymore},
  author={Robin Kothari and Ryan O'Donnell and Kewen Wu},
  journal={arXiv:2510.07515},
  year={2025}
}

@article{quantumEquivalenceSLWEISIS,
  title={On the Quantum Equivalence between $\mathrm{S}|\mathrm{LWE}\rangle$ and $\mathrm{ISIS}$},
  author={Andr{\'e} Chailloux and Paul Hermouet},
  journal={Cryptology ePrint Archive},
  year={2025}
}

@article{R04,
  title={Quantum computation and lattice problems},
  author={Oded Regev},
  journal={SIAM Journal on Computing},
  volume={33},
  number={3},
  pages={738--760},
  year={2004},
  publisher={SIAM}
}

@article{R09,
  title={On lattices, learning with errors, random linear codes, and cryptography},
  author={Oded Regev},
  journal={Journal of the ACM (JACM)},
  volume={56},
  number={6},
  pages={1--40},
  year={2009},
  publisher={ACM New York, NY, USA}
}

@Article{AR05,
  title = {Lattice problems in {NP $\cap$ coNP}},
  author = {Dorit Aharonov and Oded Regev},
  journal = {Journal of the ACM (JACM)},
  volume = {52},
  number = {5},
  pages = {749--765},
  year = {2005}
}

@article{khattar2025verifiable,
  title={Verifiable Quantum Advantage via Optimized {DQI} Circuits},
  author={Tanuj Khattar and Noah Shutty and Craig Gidney and Adam Zalcman and Noureldin Yosri and Dmitri Maslov and Ryan Babbush and Stephen P. Jordan},
  journal={arXiv:2510.10967},
  year={2025}
}

@article{buzet2025fine,
  title={Fine-Grained Unambiguous Measurements},
  author={Quentin Buzet and Andr{\'e} Chailloux},
  journal={arXiv:2510.07298},
  year={2025}
}

@article{renes2017belief,
  title={Belief propagation decoding of quantum channels by passing quantum messages},
  author={Joseph M. Renes},
  journal={New Journal of Physics},
  volume={19},
  number={7},
  pages={072001},
  year={2017},
  publisher={IOP Publishing}
}

@inproceedings{debris2024quantum,
  title={Quantum oblivious {LWE} sampling and insecurity of standard model lattice-based snarks},
  author={Thomas Debris-Alazard and Pouria Fallahpour and Damien Stehl{\'e}},
  booktitle={Proceedings of the 56th Annual ACM Symposium on Theory of Computing},
  pages={423--434},
  year={2024}
}

@inproceedings{brandsen2022belief,
  title={Belief propagation with quantum messages for symmetric classical-quantum channels},
  author={Sarah Brandsen and Avijit Mandal and Henry D. Pfister},
  booktitle={2022 IEEE Information Theory Workshop (ITW)},
  pages={494--499},
  year={2022},
  organization={IEEE}
}

@article{mandal2026belief,
  title={Belief Propagation with Quantum Messages for Symmetric Q-ary Pure-State Channels},
  author={Avijit Mandal and Henry D. Pfister},
  journal={arXiv:2601.21330},
  year={2026}
}

@Article{RU01,
  title = {The capacity of low-density parity-check codes under message-passing decoding},
  author = {Thomas J. Richardson and R{\"u}diger L. Urbanke},
  journal = {IEEE Transactions on Information Theory},
  volume = {47},
  number = {2},
  pages = {599--618},
  year = {2001}
}

@inproceedings{ATS03,
  title={Adiabatic quantum state generation and statistical zero knowledge},
  author={Dorit Aharonov and Amnon Ta-Shma},
  booktitle={Proceedings of the thirty-fifth annual ACM symposium on Theory of computing},
  pages={20--29},
  year={2003}
}

@Book{NC00,
  author =       {Michael A. Nielsen and Isaac L. Chuang},
  title =        {Quantum Computation and Quantum Information},
  publisher =    {Cambridge University Press},
  year =         {2000}
}

@article{schmidhuber2025hamiltonian,
  title={Hamiltonian decoded quantum interferometry},
  author={Alexander Schmidhuber and Jonathan Z. Lu and Noah Shutty and Stephen Jordan and Alexander Poremba and Yihui Quek},
  journal={arXiv:2510.07913},
  year={2025}
}

@article{farhi2001quantum,
  title={A quantum adiabatic evolution algorithm applied to random instances of an {NP}-complete problem},
  author={Edward Farhi and Jeffrey Goldstone and Sam Gutmann and Joshua Lapan and Andrew Lundgren and Daniel Preda},
  journal={Science},
  volume={292},
  number={5516},
  pages={472--475},
  year={2001},
  publisher={American Association for the Advancement of Science}
}

@article{farhi2014quantum,
  title={A quantum approximate optimization algorithm},
  author={Edward Farhi and Jeffrey Goldstone and Sam Gutmann},
  journal={arXiv:1411.4028},
  year={2014}
}

@article{montanaro2020quantum,
  title={Quantum speedup of branch-and-bound algorithms},
  author={Ashley Montanaro},
  journal={Physical Review Research},
  volume={2},
  number={1},
  pages={013056},
  year={2020},
  publisher={APS}
}

@article{li2025new,
  title={A New Quantum Linear System Algorithm Beyond the Condition Number and Its Application to Solving Multivariate Polynomial Systems},
  author={Jianqiang Li},
  journal={arXiv:2510.05588},
  year={2025}
}

@Article{BF21,
  author = {Joao Basso and Edward Farhi and Kunal Marwaha and Benjamin Villalonga and Leo Zhou},
  title = {The quantum approximate optimization algorithm at high depth for {MaxCut} on large-girth regular graphs and the {S}herrington-{K}irkpatrick model},
  journal = {arXiv:2110.14206},
  year = {2021}
}

@Article{FJ24,
  author = {Edward Farhi and Stephen P. Jordan},
  title = {Efficiently constructing a quantum uniform superposition over bit strings near a binary linear code},
  journal = {Quantum Information and Computation},
  volume = {24},
  number = {15/16},
  pages = {1326--1355},
  note =  {arXiv:2404.16129},
  year = {2025}
}

@misc{bensgithub,
  author = {Joao Basso and Edward Farhi and Kunal Marwaha and Benjamin Villalonga and Leo Zhou},
  title = {Performance of the {QAOA} on {MaxCut} over Large-Girth Regular Graphs},
  year = {2022},
  howpublished = {\url{https://github.com/benjaminvillalonga/large-girth-maxcut-qaoa/blob/main/data.csv}},
  commit = {b5bbc2}
}

@article{farhi2025lower,
  title={Lower bounding the {MaxCut} of high girth 3-regular graphs using the {QAOA}},
  author={Edward Farhi and Sam Gutmann and Daniel Ranard and Benjamin Villalonga},
  journal={arXiv:2503.12789},
  year={2025}
}

@article{kramer2026tight,
  title={Tight inapproximability of {max-LINSAT} and implications for decoded quantum interferometry},
  author={Maximilian J. Kramer and Carsten Schubert and Jens Eisert},
  journal={arXiv:2603.04540},
  year={2026}
}

@article{anschuetz2025decoded,
  title={Decoded quantum interferometry requires structure},
  author={Eric R. Anschuetz and David Gamarnik and Jonathan Z. Lu},
  journal={arXiv:2509.14509},
  year={2025}
}

@article{parekh2025no,
  title={No quantum advantage in decoded quantum interferometry for {MaxCut}},
  author={Ojas Parekh},
  journal={arXiv:2509.19966},
  year={2025}
}

@Article{ZP75,
  author = {Victor Vasilievich Zyablov and Mark Semenovich Pinsker},
  title = {Estimation of the error-correction complexity for {G}allager low-density codes},
  journal = {Problemy Peredachi Informatsii},
  volume = {11},
  number = {1},
  pages = {23--36},
  year = {1975},
}

@unpublished{zhou2026gibbs,
  title  = {Gibbs sampling by Decoded Quantum Interferometry},
  author = {Leo Zhou and Noah Shutty and Mark Sellke and Stephen P. Jordan},
  note   = {In preparation},
  year   = {2026}
}

\end{document}